\documentclass[journal,12pt,onecolumn,draftclsnofoot,]{IEEEtran}
\usepackage{fontenc}
\usepackage{graphicx}
\usepackage[cmex10]{amsmath}
\usepackage{array}
\usepackage{amssymb}
\usepackage{amsthm}
\usepackage{xcolor}
\usepackage[FIGBOTCAP]{subfigure}
\usepackage{parskip}
\usepackage{mathrsfs}
\usepackage{bm}
\usepackage[noadjust]{cite}
\usepackage{amsmath,graphicx}
\usepackage[bookmarks=false]{hyperref}
\usepackage[font=small,labelfont=bf]{caption}
\hypersetup{
    colorlinks=true,
    linkcolor=blue,
    filecolor=magenta,
    urlcolor=blue,
}
\usepackage{bm}
\usepackage{epstopdf}
\usepackage{multirow}
\usepackage{algorithm}
\usepackage{algpseudocode}
\usepackage{amssymb}
\usepackage{amsmath}
\usepackage{fancyhdr}
\usepackage{graphicx}
\usepackage{float}
\usepackage{subfigure}
\usepackage{comment}

\theoremstyle{definition}
\newtheorem{definition}{Definition}
\newtheorem{proposition}{Proposition}
\newtheorem{corollary}{Corollary}

\begin{document}

\graphicspath{{figures/}}

\title{Multivariate Extension of Matrix-based R{\'e}nyi's $\alpha$-order Entropy Functional}
\author{Shujian~Yu,~\IEEEmembership{Student~Member,~IEEE}, Luis~Gonzalo~S\'{a}nchez~Giraldo,\\
Robert Jenssen,~\IEEEmembership{Member,~IEEE}, and Jos\'{e}~C.~Pr\'{i}ncipe,~\IEEEmembership{Fellow,~IEEE}%
\thanks{S.~Yu and J.~C.~Pr\'{i}ncipe are with the Department of Electrical and Computer Engineering, University of Florida, Gainesville, FL 32611,
  USA (email: yusjlcy9011@ufl.edu; principe@cnel.ufl.edu).}
\thanks{L.~G.~S\'{a}nchez Giraldo is with the Department of Computer Science, University of Miami, Coral Gables, FL 33124,
USA (email: lgsanchez@cs.miami.edu).}
\thanks{R.~Jenssen is with the Department of Physics and Technology, UiT - The Arctic University of Norway, Troms{\o} 9037, Norway (email: robert.jenssen@uit.no).}}%

\IEEEtitleabstractindextext{%
	\begin{abstract}
	  The matrix-based R{\'e}nyi's $\alpha$-order entropy functional was recently introduced using the normalized eigenspectrum of a Hermitian matrix of the projected data in a reproducing kernel Hilbert space (RKHS). However, the current theory in the matrix-based R{\'e}nyi's $\alpha$-order entropy functional only defines the entropy of a single variable or mutual information between two random variables. In information theory and machine learning communities, one is also frequently interested in multivariate information quantities, such as the multivariate joint entropy and different interactive quantities among multiple variables. In this paper, we first define the matrix-based R{\'e}nyi's $\alpha$-order joint entropy among multiple variables. We then show how this definition can ease the estimation of various information quantities that measure the interactions among multiple variables, such as interactive information and total correlation. We finally present an application to feature selection to show how our definition provides a simple yet powerful way to estimate a widely-acknowledged intractable quantity from data. A real example on hyperspectral image (HSI) band selection is also provided.
	\end{abstract}
	
	}

\maketitle
\IEEEdisplaynontitleabstractindextext
\IEEEpeerreviewmaketitle
\graphicspath{{figures/}}

%

\section{Introduction}\label{sec:introduction}
The R{\'e}nyi's $\alpha$-order entropy~\cite{renyi1961measures} was defined in $1961$ as a one-parameter generalization of the celebrated Shannon entropy. In the same paper, Alfr{\'e}d R{\'e}nyi also introduced the $\alpha$-order divergence as a natural extension of the Shannon relative entropy. Following R{\'e}nyi's work, different definitions on $\alpha$-order mutual information have been proposed in the last decades, demonstrating elegant properties and great potentials for widespread adoption~\cite{verdu2015alpha}.

Fifty years after the definition of Alfr{\'e}d R{\'e}nyi, the matrix-based R{\'e}nyi's $\alpha$-order entropy functional was introduced by S\'{a}nchez Giraldo \emph{et al.}~\cite{giraldo2015measures}, in which both the entropy and the mutual information are defined over the normalized eigenspectrum of the Gram matrix, which is a Hermitian matrix of the projected data in a reproducing kernel Hilbert space (RKHS). These new functional definitions do not require a probability interpretation and avoid real-valued or discrete probability density function (PDF) estimation, but exhibit similar properties to R{\'e}nyi's $\alpha$-order entropy.

However, the current formulations in this theory only define the entropy of a single variable or the mutual information between two variables, which can be a limiting factor when multiple variables are available. In information theory and machine learning communities, one is also frequently interested in multivariate information quantities, such as the multivariate joint entropy and the interactions among multiple variables. For example, in multi-input single-output (MISO) communication systems, the basic bivariate model (between one input and one output) will certainly fail to discriminate effects due to uncontrolled input sources from those due to random noises, i.e., we cannot figure out the impairments due to system noise in the absence of knowledge of the relationships among multiple input sources~\cite{timme2011multivariate}. In machine learning, we are always interested in measuring the relationships among two or more variables to enable learning more compact representations~\cite{ver2016information} or for selecting a more informative feature set~\cite{brown2012conditional}.


In this paper, we extend S\'{a}nchez Giraldo \emph{et~al.}'s definition to the multivariate scenario and illustrate the characteristics and potential applications of this extension. Specifically, in section~\ref{section2}, we provide the definitions of the matrix-based R{\'e}nyi's $\alpha$-order entropy functional, including (joint) entropy and mutual information. Then, in section~\ref{section3}, we define the proposed extension of the matrix-based R{\'e}nyi's $\alpha$-order joint entropy to multiple variables and formally show it is consistent with the bivariate definition. After that, in section~\ref{section4}, we show that this matrix-based formulation on the normalized eigenspectrum enables straightforward definitions of interactions among multiple variables and give an example of their applicability for feature selection in section~\ref{real_application} that illustrates how this simple definition provides advantageous result in comparison to well known techniques. We finally conclude this paper and provide an outlook regarding the potential of our definitions for future work in section~\ref{conclusions}.

\section{Preliminary knowledge: from Renyi's entropy to its matrix-based functional} \label{section2}

In information theory, a natural extension of the well-known Shannon's entropy is R{\'e}nyi's $\alpha$-order entropy~\cite{renyi1961measures}. For a random variable $X$ with probability density function (PDF) $f(x)$ in a finite set $\mathcal{X}$, the $\alpha$-entropy $\mathbf{H}_\alpha(X)$ is defined as:

\begin{equation}
\mathbf{H}_{\alpha}(f)=\frac{1}{1-\alpha}\log\int_\mathcal{X}f^\alpha(x)dx\label{eq1}
\end{equation}

The limiting case of Eq.~(\ref{eq1}) for $\alpha\rightarrow1$ yields Shannon's differential entropy. It also turns out that for any positive real $\alpha$, the above quantity can be expressed, under some restrictions, as a function of inner products between PDFs~\cite{principe2010information}. In particular, the $2$-order entropy of $f$ and the cross-entropy between $f$ and $g$ along with Parzen density estimation~\cite{parzen1962estimation} yield simple yet elegant expressions that can serve as objective functions for a family of supervised or unsupervised learning algorithms when the PDF is unknown~\cite{principe2010information}.

R{\'e}nyi's entropy and divergence evidence a long track record of usefulness in information theory and its applications~\cite{principe2010information}. Unfortunately, the accurate PDF estimation of high dimensional, continuous, and complex data impedes its more widespread adoption in data-driven science. To solve this problem, S\'{a}nchez Giraldo $et~al$.~\cite{giraldo2015measures} suggested a quantity that resembles quantum R{\'e}nyi's entropy~\cite{muller2013quantum} defined in terms of the normalized eigenspectrum of the Gram matrix of the data projected to an RKHS, thus estimating the entropy directly from data without PDF estimation. S\'{a}nchez Giraldo $et~al$.'s matrix entropy functional is defined as follows.

\theoremstyle{definition}
\begin{definition} Let $\kappa:\mathcal{X}\times\mathcal{X}\mapsto\mathbb{R}$ be a real valued positive definite kernel that is also infinitely divisible~\cite{bhatia2006infinitely}. Given $X=\{x_1,x_2,...,x_n\}$ and the Gram matrix $K$ obtained from evaluating a positive definite kernel $\kappa$ on all pairs of exemplars, that is $(K)_{ij}=\kappa(x_i,x_j)$, a matrix-based analogue to R{\'e}nyi's $\alpha$-entropy for a normalized positive definite (NPD) matrix $A$ of size $n\times n$,  such that $\mathrm{tr}(A)=1$, can be given by the following functional:
\begin{equation}
\mathbf{S}_\alpha(A)=\frac{1}{1-\alpha}\log_2\left(\mathrm{tr}(A^\alpha)\right)=
\frac{1}{1-\alpha}\log_2\big[\sum_{i=1}^n\lambda_i(A)^\alpha\big] \label{eq8}
\end{equation}
where $A_{ij}=\frac{1}{n}\frac{K_{ij}}{\sqrt{K_{ii}K_{jj}}}$ and $\lambda_i(A)$ denotes the $i$-th eigenvalue of $A$.
\end{definition}

\theoremstyle{definition}
\begin{definition}
Given $n$ pairs of samples $\{z_i=(x_i,y_i)\}_{i=1}^n$, each sample contains two different types of measurements $x\in \mathcal{X}$ and $y\in \mathcal{Y}$ obtained from the same realization, and the positive definite kernels $\kappa_1:\mathcal{X}\times \mathcal{X}\mapsto\mathbb{R}$ and $\kappa_2:\mathcal{Y}\times \mathcal{Y}\mapsto\mathbb{R}$, a matrix-based analogue to R{\'e}nyi's $\alpha$-order joint-entropy can be defined as:
\begin{equation}
\mathbf{S}_\alpha(A,B)=\mathbf{S}_\alpha\left(\frac{A\circ B}{\mathrm{tr}(A\circ B)}\right) \label{eq10}
\end{equation}
where $A_{ij}=\kappa_1(x_i,x_j)$, $B_{ij}=\kappa_2(x_i,x_j)$ and $A\circ B$ denotes the Hadamard product between the matrices $A$ and $B$. The local structure of the Gram matrices $A$ and $B$ simplifies the estimation of the joint distribution to pairwise element multiplication and it is the source of the simplicity of our estimation methodology.
\end{definition}

The following proposition proved by S\'{a}nchez Giraldo, $et~al$.~\cite[page~5]{giraldo2015measures} makes the definition of the above joint entropy compatible with the individual entropies of its components, and also allows us to define a matrix notion of R{\'e}nyi's conditional entropy $\mathbf{S}_{\alpha}(A|B)$ (or $\mathbf{S}_{\alpha}(B|A)$) and mutual information $\mathbf{I}_{\alpha}(A;B)$ in analogy with Shannon's definition.

\begin{proposition}
Let $A$ and $B$ be two $n\times n$ positive definite matrices with trace $1$ with nonnegative entries, and $A_{ii}=B_{ii}=\frac{1}{n}$, for $i=1,2,\cdots,n$. Then the following two inequalities hold:
\vspace{-5pt}
\begin{enumerate}
\item $\mathbf{S}_\alpha\big(\frac{A\circ B}{\mathrm{tr}(A\circ B)}\big)\leq\mathbf{S}_\alpha(A)+\mathbf{S}_\alpha(B),$
\item $\mathbf{S}_\alpha\big(\frac{A\circ B}{\mathrm{tr}(A\circ B)}\big)\geq\max[\mathbf{S}_\alpha(A),\mathbf{S}_\alpha(B)].$
\end{enumerate}
\end{proposition}
\vspace{-10pt}
Since there is no consensus on the definition of R{\'e}nyi's conditional entropy and mutual information~\cite{verdu2015alpha}, motivated by the additive and subtractive relationships among different information theoretic quantities of Shannon's definition, $\mathbf{S}_\alpha(A|B)$ and $\mathbf{I}_\alpha(A;B)$ can be computed as:
\begin{equation}
\mathbf{S}_\alpha(A|B)=\mathbf{S}_\alpha(A,B)-\mathbf{S}_\alpha(B),\label{eq5}
\end{equation}
\begin{equation}
\mathbf{I}_\alpha(A;B)=\mathbf{S}_\alpha(A)+\mathbf{S}_\alpha(B)-\mathbf{S}_\alpha(A,B). \label{eq6}
\end{equation}

In this paper, we use the radial basis function (RBF) kernel $\kappa(x_i,x_j)=\exp(-\frac{\|x_i-x_j\|^2}{2\sigma^2})$ to obtain the Gram matrices. This way, the user has to make two decisions (hyper-parameters) that change with the data and the task goal: the selection of the kernel size $\sigma$ to project the data to the RKHS and the selection of the order $\alpha$. The selection of $\sigma$ can follow Silverman's rule of thumb for density estimation~\cite{silverman1986density}, or other heuristics from a graph cut perspective, such as $10$ to $20$ percent of the total range of the Euclidean distances between all pairwise data points~\cite{shi2000normalized}. The choice of $\alpha$ is associated with the task goal. If the application requires emphasis on tails of the distribution (rare events) or multiple modalities, $\alpha$ should be less than $2$ and possibly approach to $1$ from above. $\alpha=2$ provides neutral weighting~\cite{principe2010information}. Finally, if the goal is to characterize modal behavior, $\alpha$ should be greater than $2$.

\section{Joint entropy among multiple variables} \label{section3}
In this section, we first give the definition of the matrix-based R{\'e}nyi's $\alpha$-order joint-entropy among multiple variables and then present two corollaries that serve as a foundation to this definition.

\theoremstyle{definition}
\begin{definition}
Given a collection of $n$ samples $\{s_i=(x_1^i,x_2^i,\cdots, x_k^i)\}_{i=1}^n$, where the superscript $i$ denotes the sample index, each sample contains $k$ ($k\geq2$) measurements $x_1\in \mathcal{X}_1$, $x_2\in \mathcal{X}_2$, $\cdots$, $x_k\in \mathcal{X}_k$ obtained from the same realization, and the positive definite kernels $\kappa_1:\mathcal{X}_1\times \mathcal{X}_1\mapsto\mathbb{R}$, $\kappa_2:\mathcal{X}_2\times \mathcal{X}_2\mapsto\mathbb{R}$, $\cdots$, $\kappa_k:\mathcal{X}_k\times \mathcal{X}_k\mapsto\mathbb{R}$, a matrix-based analogue to R{\'e}nyi's $\alpha$-order joint-entropy among $k$ variables can be defined as:
\begin{equation}
\mathbf{S}_\alpha(A_1,A_2,\cdots,A_k)=\mathbf{S}_\alpha\left(\frac{A_1\circ A_2\circ\cdots\circ A_k}{\mathrm{tr}(A_1\circ A_2\circ\cdots\circ A_k)}\right) \label{eq7}
\end{equation}
where $(A_1)_{ij}=\kappa_1(x_1^i,x_1^j)$, $(A_2)_{ij}=\kappa_2(x_2^i,x_2^j)$, $\cdots$, $(A_k)_{ij}=\kappa_k(x_k^i,x_k^j)$, and $\circ$ denotes the Hadamard product.
\end{definition}

The following two corollaries provide the theoretical backing for using (\ref{eq7}) to quantify joint entropy among multiple variables.

\begin{corollary}
Let $[k]$ be the index set $\{1,2,\cdots,k\}$. We partition $[k]$ into two complementary subsets $\mathbf{s}$ and $\tilde{\mathbf{s}}$. For any $\mathbf{s}\subset[k]$, denote all indices in $\mathbf{s}$ with $s_1, s_2, \cdots, s_{|\mathbf{s}|}$, where $|\cdot|$ stands for cardinality. Similarly, denote all indices in $\tilde{\mathbf{s}}$ with $\tilde{s}_1, \tilde{s}_2, \cdots, \tilde{s}_{|\mathbf{\tilde{s}}|}$. Also let $A_1$, $A_2$, $\cdots$, $A_k$ be $k$ $n\times n$ positive definite matrices with trace $1$ and nonnegative entries, and $(A_1)_{ii}=(A_2)_{ii}=\cdots=(A_k)_{ii}=\frac{1}{n}$, for $i=1,2,\cdots,n$. Then the following two inequalities hold:

\begin{equation} \label{corollary1.1}
\mathbf{S}_\alpha\left(\frac{A_1\circ A_2\circ\cdots\circ A_k}{\mathrm{tr}(A_1\circ A_2\circ\cdots\circ A_k)}\right)\leq
    \mathbf{S}_\alpha\left(\frac{A_{s_1}\circ A_{s_2}\circ\cdots\circ A_{s_{|\mathbf{s}|}}}{\mathrm{tr}(A_{s_1}\circ A_{s_2}\circ\cdots\circ A_{s_{|\mathbf{s}|}})}\right)+
    \mathbf{S}_\alpha\left(\frac{A_{\tilde{s}_1}\circ A_{\tilde{s}_2}\circ\cdots\circ A_{\tilde{s}_{|\mathbf{\tilde{s}}|}}}{\mathrm{tr}(A_{\tilde{s}_1}\circ A_{\tilde{s}_2}\circ\cdots\circ A_{\tilde{s}_{|\mathbf{\tilde{s}}|}})}\right),
\end{equation}

\begin{equation} \label{corollary1.2}
\mathbf{S}_\alpha\left(\frac{A_1\circ A_2\circ\cdots\circ A_k}{\mathrm{tr}(A_1\circ A_2\circ\cdots\circ A_k)}\right)\geq\max
    \left[\mathbf{S}_\alpha\left(\frac{A_{s_1}\circ A_{s_2}\circ\cdots\circ A_{s_{|\mathbf{s}|}}}{\mathrm{tr}(A_{s_1}\circ A_{s_2}\circ\cdots\circ A_{s_{|\mathbf{s}|}})}\right),
    \mathbf{S}_\alpha\left(\frac{A_{\tilde{s}_1}\circ A_{\tilde{s}_2}\circ\cdots\circ A_{\tilde{s}_{|\mathbf{\tilde{s}}|}}}{\mathrm{tr}(A_{\tilde{s}_1}\circ A_{\tilde{s}_2}\circ\cdots\circ A_{\tilde{s}_{|\mathbf{\tilde{s}}|}})}\right)\right].
\end{equation}

\end{corollary}

\begin{proof}
Setting $A=\frac{A_{s_1}\circ A_{s_2}\circ\cdots\circ A_{s_{|\mathbf{s}|}}}{\mathrm{tr}(A_{s_1}\circ A_{s_2}\circ\cdots\circ A_{s_{|\mathbf{s}|}})}$ and $B=\frac{A_{\tilde{s}_1}\circ A_{\tilde{s}_2}\circ\cdots\circ A_{\tilde{s}_{|\mathbf{\tilde{s}}|}}}{\mathrm{tr}(A_{\tilde{s}_1}\circ A_{\tilde{s}_2}\circ\cdots\circ A_{\tilde{s}_{|\mathbf{\tilde{s}}|}})}$. According to the Schur product theorem, $A$ and $B$ are $n\times n$ positive definite matrices with trace $1$ and nonnegative entries, and $A_{ii}=B_{ii}=\frac{1}{n}$, for $i=1,2,\cdots,n$. Then \textbf{Proposition $1$} implies (\ref{corollary1.1}) and (\ref{corollary1.2}).
\end{proof}

\begin{corollary}
Let $A_1$, $A_2$, $\cdots$, $A_k$ be $k$ $n\times n$ positive definite matrices with trace $1$ and nonnegative entries, and $(A_1)_{ii}=(A_2)_{ii}=\cdots=(A_k)_{ii}=\frac{1}{n}$, for $i=1,2,\cdots,n$. Then the following two inequalities hold:

\begin{equation} \label{corollary2.1}
\mathbf{S}_\alpha\left(\frac{A_1\circ A_2\circ\cdots\circ A_k}{\mathrm{tr}(A_1\circ A_2\circ\cdots\circ A_k)}\right)\leq
\mathbf{S}_\alpha(A_1)+\mathbf{S}_\alpha(A_2)+\cdots+\mathbf{S}_\alpha(A_k),
\end{equation}

\begin{equation} \label{corollary2.2}
\mathbf{S}_\alpha\left(\frac{A_1\circ A_2\circ\cdots\circ A_k}{\mathrm{tr}(A_1\circ A_2\circ\cdots\circ A_k)}\right)\geq\max
[\mathbf{S}_\alpha(A_1),\mathbf{S}_\alpha(A_2),\cdots,\mathbf{S}_\alpha(A_k)].
\end{equation}

\end{corollary}

\begin{proof}
For every $i\in[2,k]$, let $\mathbf{s}=\{i\}$ and $\mathbf{\tilde{s}}=\{1,2,\cdots,i-1\}$, by \textbf{Corollary $1$}, we have:
\begin{equation} \label{eq12}
\mathbf{S}_\alpha\left(\frac{A_1\circ A_2\circ\cdots\circ A_i}{\mathrm{tr}(A_1\circ A_2\circ\cdots\circ A_i)}\right)\leq
\mathbf{S}_\alpha(A_i)+
\mathbf{S}_\alpha\left(\frac{A_1\circ A_2\circ\cdots\circ A_{i-1}}{\mathrm{tr}(A_1\circ A_2\circ\cdots\circ A_{i-1})}\right),
\end{equation}

\begin{equation} \label{eq13}
\mathbf{S}_\alpha\left(\frac{A_1\circ A_2\circ\cdots\circ A_i}{\mathrm{tr}(A_1\circ A_2\circ\cdots\circ A_i)}\right)\geq
\max\left[\mathbf{S}_\alpha(A_i),
\mathbf{S}_\alpha\left(\frac{A_1\circ A_2\circ\cdots\circ A_{i-1}}{\mathrm{tr}(A_1\circ A_2\circ\cdots\circ A_{i-1})}\right)\right],
\end{equation}

Adding the $k-1$ inequalities in (\ref{eq12}) and subtracting common terms on both sides, we get (\ref{corollary2.1}). Similarly, combing the $k-1$ inequalities in (\ref{eq13}), we get (\ref{corollary2.2}).
\end{proof}

\section{Interaction quantities among multiple variables} \label{section4}

Given the definitions in section~\ref{section3}, we discuss the matrix-based analogues to three multivariate information quantities that were introduced in previous work to measure the interactions among multiple variables. Note that, there are various definitions to measure such interactions. Here, we only review three of the major ones, as this section aims to illustrate the great simplicity offered by our definitions. Interested readers can refer to \cite{timme2011multivariate} for an experimental survey on different definitions and their properties.

\subsection{Mutual information}
The mutual information can be extended straightforwardly as a measure of the interactions among more than two variables by grouping the variables into sets, treating each set as a new single variable. For instance, the total amount of information about a random variable $Y$ that is gained from the other $k$ variables, $X_1$, $X_2$, $\cdots$, $X_k$, can be defined as:
\begin{multline}
\mathbf{I}_\alpha(B;\{A_1,A_2,\cdots,A_k\})=\mathbf{S}_\alpha(B)+\\ \mathbf{S}_\alpha\left(\frac{A_1\circ A_2\circ\cdots\circ A_k}{\mathrm{tr}(A_1\circ A_2\circ\cdots\circ A_k)}\right)
-\mathbf{S}_\alpha\left(\frac{A_1\circ A_2\circ\cdots\circ A_k\circ B}{\mathrm{tr}(A_1\circ A_2\circ\cdots\circ A_k\circ B)}\right), \label{eq14}
\end{multline}
where $A_1$, $A_2$, $\cdots$, $A_k$, and $B$ denote the normalized Gram matrices evaluated over $X_1$, $X_2$, $\cdots$, $X_k$, and $Y$ respectively. According to \textbf{Corollary $1$}, $\mathbf{I}_\alpha(B;{A_1,A_2,…,A_k})\geq0$. However, Eq.~(\ref{eq14}) cannot measure separately contributions in the information about $Y$ from individual variables $X_i$ for $i=1,2,\cdots,k$.

\subsection{Interaction information}

Interaction information (II) \cite{mcgill1954multivariate} extends the concept of the mutual information as the information gained about one variable by knowing the other \cite{timme2011multivariate}. This way, the II among three variables is defined as the gain (or loss) in sample information transmitted between any two of the variables, due to the additional knowledge of a third variable~\cite{mcgill1954multivariate}:
\begin{equation}\label{eq_mcgill}
\begin{split}
\mathbf{II}(X_1;X_2;X_3)&=\mathbf{I}(X_1;X_2|X_3)-\mathbf{I}(X_1;X_2) \\
&=\mathbf{I}(\{X_1,X_2\};X_3)-\mathbf{I}(X_1;X_3)-\mathbf{I}(X_1;X_3).
\end{split}
\end{equation}

Eq.~(\ref{eq_mcgill}) can be written as an expansion of the entropies and joint entropies of the variables,
\begin{multline}\label{eq_mcgill_expansion}
\mathbf{II}(X_1;X_2;X_3)=-[\mathbf{H}(X_1)+\mathbf{H}(X_2)+\mathbf{H}(X_3)]\\
+[\mathbf{H}(X_1,X_2)+\mathbf{H}(X_1,X_3)+\mathbf{H}(X_2,X_3)]-\mathbf{H}(X_1,X_2,X_3).
\end{multline}

This form leads to an expansion of the II to $k$ number of variables (i.e., $k$-way interactions). Given $\mathcal{S}=\{X_1,X_2,...,X_k\}$, let $\mathcal{T}$ denote a subset of $\mathcal{S}$, then the II becomes an alternating sum over all subsets $\mathcal{T}\subseteq \mathcal{S}$~\cite{jakulin2003quantifying}:
\begin{equation}\label{eq_II_expansion}
\mathbf{II}(\mathcal{S})=-\sum\limits_{\mathcal{T}\subseteq \mathcal{S}}(-1)^{|\mathcal{S}|-|\mathcal{T}|}\mathbf{H}(\mathcal{T}).
\end{equation}

A similar quantity to II is the co-information (CI)~\cite{bell2003co}, which can be derived using a lattice structure of statistical dependency~\cite{bell2003co}. Specifically, CI is expressed as:
\begin{equation}\label{eq_CI_expansion}
\mathbf{CI}(\mathcal{S})=-\sum\limits_{\mathcal{T}\subseteq \mathcal{S}}(-1)^{|\mathcal{T}|}\mathbf{H}(\mathcal{T})=(-1)^{|\mathcal{S}|}\mathbf{II}(\mathcal{S}).
\end{equation}
Clearly, CI is equal to II except for a change in sign in the case that $\mathcal{S}$ contains an odd number of variables. Compared to II, CI ensures a proper set- or measure-theoretic interpretation: CI measures the centermost atom to which all variables contribute when we use Venn Diagrams to represent different entropy terms~\cite{fano1961transmission,yeung1991new}. Note that, the difference in the sign also gives different meanings to CI and II. For example, a positive value implies redundancy for CI, but synergy for II~\cite{timme2011multivariate}.


Given Eqs.~(\ref{eq_II_expansion}) and (\ref{eq_CI_expansion}), let $[k]$ be the index set $\{1,2,\cdots,k\}$, for any $\mathbf{s}\subseteq[k]$, denote all indices in $\mathbf{s}$ by $s_1$, $s_2$, $\cdots$, $s_{|\mathbf{s}|}$, where $|\cdot|$ stands for cardinality, the matrix-form of II and CI can be measured with Eq.~(\ref{eq15}) and Eq.~(\ref{eq16}), respectively:
\begin{equation}
\mathbf{II}_\alpha(A_1;A_2;\cdots;A_k)=-\sum\limits_{\mathbf{s}\subseteq[k]}(-1)^{k-|\mathbf{s}|}
\mathbf{S}_\alpha\left(\frac{A_{s_1}\circ A_{s_2}\circ\cdots\circ A_{s_{|\mathbf{s}|}}}{\mathrm{tr}(A_{s_1}\circ A_{s_2}\circ\cdots\circ A_{s_{|\mathbf{s}|}})}\right). \label{eq15}
\end{equation}
\begin{equation}
\mathbf{CI}_\alpha(A_1;A_2;\cdots;A_k)=-\sum\limits_{\mathbf{s}\subseteq[k]}(-1)^{|\mathbf{s}|}
\mathbf{S}_\alpha\left(\frac{A_{s_1}\circ A_{s_2}\circ\cdots\circ A_{s_{|\mathbf{s}|}}}{\mathrm{tr}(A_{s_1}\circ A_{s_2}\circ\cdots\circ A_{s_{|\mathbf{s}|}})}\right). \label{eq16}
\end{equation}


\subsection{Total correlation}
The total correlation (TC)~\cite{watanabe1960information} is defined by extending the idea that mutual information is the KL divergence between the joint distribution and the product of marginals. It measures the total amount of dependence among the variables. Formally, TC can be written in terms of individual entropies and joint entropy as:

\small
\begin{equation}
\begin{split}
\mathbf{TC}_\alpha(A_1,A_2,\cdots,A_k)&=\mathbf{S}_\alpha(A_1)+\mathbf{S}_\alpha(A_2)+\cdots+
\mathbf{S}_\alpha(A_k)\\
&-\mathbf{S}_\alpha\big(\frac{A_1\circ A_2\circ\cdots\circ A_k}{\mathrm{tr}(A_1\circ A_2\circ\cdots\circ A_k)}\big). \label{eq17}
\end{split}
\end{equation}
\normalsize
By \textbf{Corollary $2$}, $\mathbf{TC}_\alpha(A_1,A_2,\cdots,A_k)\geq0$. TC is zero if and only if all variables are mutually independent~\cite{timme2011multivariate}.

\section{Application for feature selection} \label{real_application}

In sections~\ref{section3} and~\ref{section4}, we generalized the matrix-based R{\'e}nyi's $\alpha$-order joint entropy to multiple variables. The new definition enables efficient and effective measurement of various multivariate interaction quantities. With these novel definitions, we are ready to address the problem of feature selection. Given a set of variables $S=\{X_1,X_2,\cdots,X_n\}$, feature selection refers to seeking a small subset of informative variables $S^*\subset S$ from $S$, such that the subset $S^*$ contains the most relevant yet least redundant information about a desired variable $Y$.

Suppose we want to select $k$ variables, then the ultimate objective is to maximize $\mathbf{I}(\{X_{i_1},X_{i_2},\cdots,X_{i_k}\};Y)$, where $i_1$, $i_2$, $\cdots$, $i_k$ denote the indices of selected variables. Despite the simple expression, before our work, the estimation of this quantity was considered intractable~\cite{fleuret2004fast,brown2012conditional}, even with the aid of Shannon's chain rule~\cite{mackay2003information}. As a result, tremendous efforts have been made to use different information-theoretic criteria to approximate $\mathbf{I}(\{X_{i_1},X_{i_2},\cdots,X_{i_k}\};Y)$ by retaining only the first-order or at most the second-order interactions terms amongst different features~\cite{brown2012conditional}. The theoretical relation amongst different criteria in different methods was recently investigated by Brown~\emph{et al.}~\cite{brown2012conditional}. According to the authors, numerous criteria proposed in the last decades can be placed under the same umbrella, i.e., balancing the tradeoff using different assumptions among three key terms: the individual predictive power of the feature, the unconditional correlations and the class-conditional correlations.



Obviously, benefitting from the novel definition proposed in this paper, we can now explicitly maximize the ultimate objective $\mathbf{I}(\{X_{i_1},X_{i_2},\cdots,X_{i_k}\};Y)$, without any approximations or decompositions, using Eq.~(\ref{eq14}). We compare our method with $6$ state-of-the-art information-theoretic feature selection methods, namely \emph{Mutual Information-based Feature Selection} (MIFS)~\cite{battiti1994using}, \emph{First-Order Utility} (FOU)~\cite{brown2009new}, \emph{Mutual Information Maximization} (MIM)~\cite{lewis1992feature}, \emph{Maximum-Relevance Minimum-Redundancy} (MRMR)~\cite{peng2005feature}, \emph{Joint Mutual Information} (JMI)~\cite{yang2000data} and \emph{Conditional Mutual Information Maximization} (CMIM)~\cite{fleuret2004fast}. Among them, MIM is the baseline method that scores each feature independently without considering any interactions terms. MRMR is perhaps the most widely used method in various applications. According to~\cite{brown2012conditional}, JMI and CMIM outperform their counterparts, since both methods integrate the aforementioned three key terms.

We list the criteria for all methods in Table~\ref{lab:difference_summarization} for clarity. All the methods employ a greedy procedure to incrementally build the selected feature set, in each step. We implemented and optimized the codes for all the above methods in Matlab $2016$b. All methods are compared in terms of the average cross validation (CV) classification accuracy on a range of features. We employ $10$-fold CV in datasets with sample size more than $100$ and leave-one-out (LOO) CV otherwise. One should also note that the majority of the prevalent information-theoretic feature selection methods are built upon classic discrete Shannon's information quantities~\cite{li2017feature}. For mutual information estimation of those methods, continuous features are discretized using an equal-width strategy into $5$ bins~\cite{vinh2014reconsidering}, while features already with a categorical range were left untouched. For our method, we use $\alpha=1.01$ as suggested in~\cite{giraldo2015measures} to approximate the Shannon's information and make this comparison fair. We also fix the kernel size $\sigma=1$ in all experiments for simplicity. A thorough treatment to effects of $\alpha$ and $\sigma$ is discussed later.

\subsection{Artificial data} \label{app_artificial}
In this first experiment, we wish to evaluate all competing methods on data in which the optimal number of features and the inter-dependencies amongst features are known in advance. To this end, we select the MADELON dataset, a well-known benchmark from the NIPS $2003$ Feature Selection Challenge~\cite{guyon2005result}. MADELON is an artificial dataset containing $2000$ data points grouped in $32$ clusters placed on the vertices of a five-dimensional hypercube and randomly labeled $+1$ or $-1$. The five dimensions constitute $5$ informative features. $15$ linear combinations of those informative features were added to form redundant informative features. Based on those informative $20$ features one must separate the examples into the $2$ classes ($\pm1$ labels). Apart from informative features, MADELON was also added $480$ distractor features (or noises) called ``probes" having no predictive power. The order of the features and patterns were randomized.

Following~\cite{guyon2005result}, we use a $3$-NN classifier and select $20$ features. Fig.~\ref{fig:real_data}(a) shows the validation results. As can be seen, our method demonstrates overwhelming advantage on the first $5$ features. Our method works very similar to the ultimate objective, as opposed to the other methods that neglect high-order interactions terms. One should also note that the advantage of our method becomes weaker after $6$ features. This is because, after selecting the $5$ most informative features, the linear combinations of informative features become redundant information to our method such that their functionalities can be fully substituted with the first $5$ features. In other words, the value of $\mathbf{I}(\{X_{i_1},X_{i_2},\cdots,X_{i_k}\};Y)-\mathbf{I}(\{X_{i_1},X_{i_2},\cdots,X_{i_{k-1}}\};Y)$ ($k>5$) becomes tiny such that our method cannot distinguish linear combinations from noises. By contrast, since other methods cannot find the most $5$ informative features at first, it is possible for them to select one of the most informative features in later steps. For example, if one method selects the combination of the first and second informative features at step $3$, this method may even achieve higher classification accuracy in later steps if the third or fourth informative feature is selected at that step. This is because of the possible existence of synergistic information~\cite{williams2010nonnegative}. In our approach this would call for smaller and smaller $\sigma$ and/or different values of $\alpha$, but it was not implemented here.

\begin{figure*}[!t]
\setlength{\abovecaptionskip}{0.cm}
\setlength{\belowcaptionskip}{-0.0cm}
\centering
\subfigure[MADELON] {\includegraphics[width=.32\textwidth]{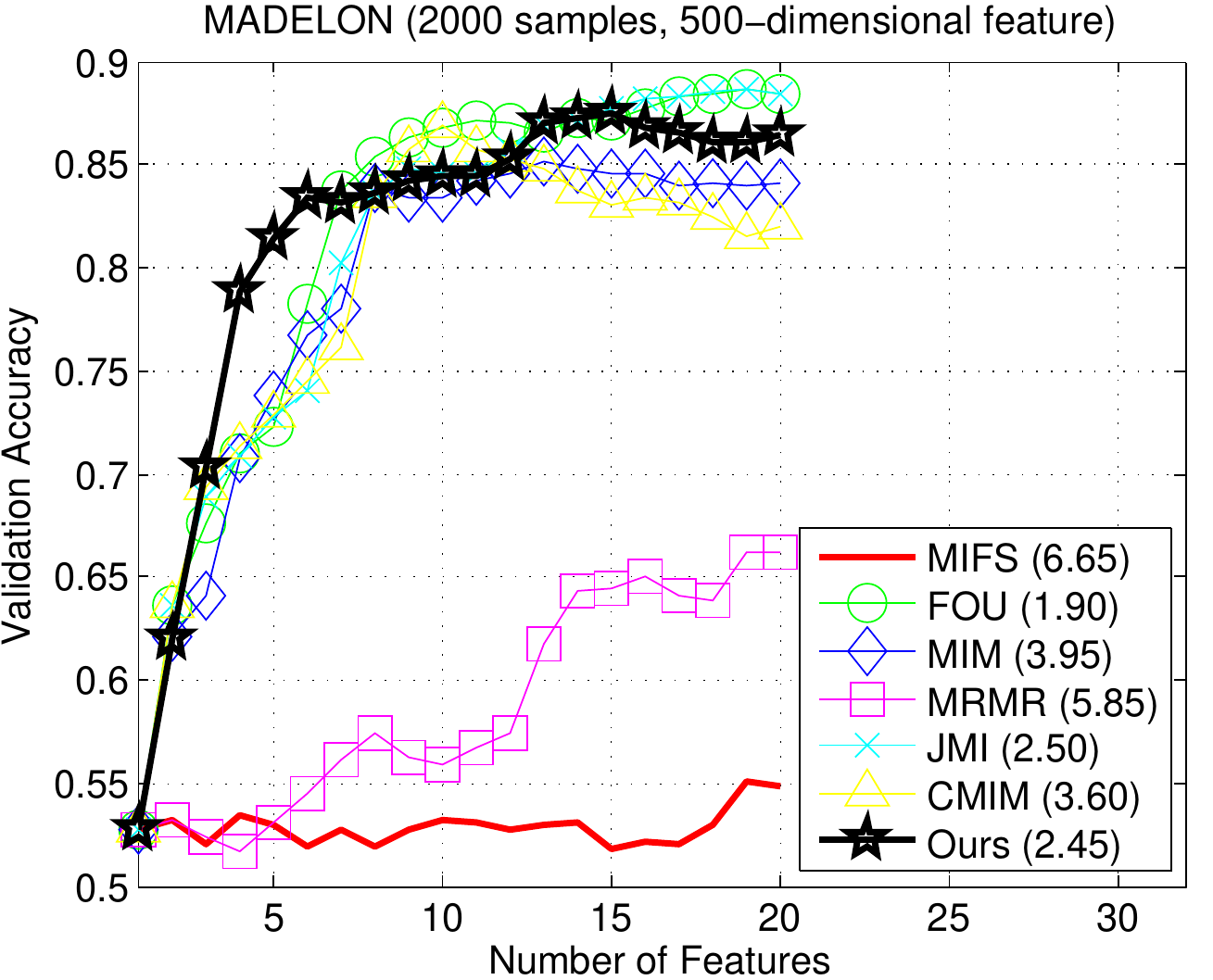}}
\subfigure[breast] {\includegraphics[width=.32\textwidth]{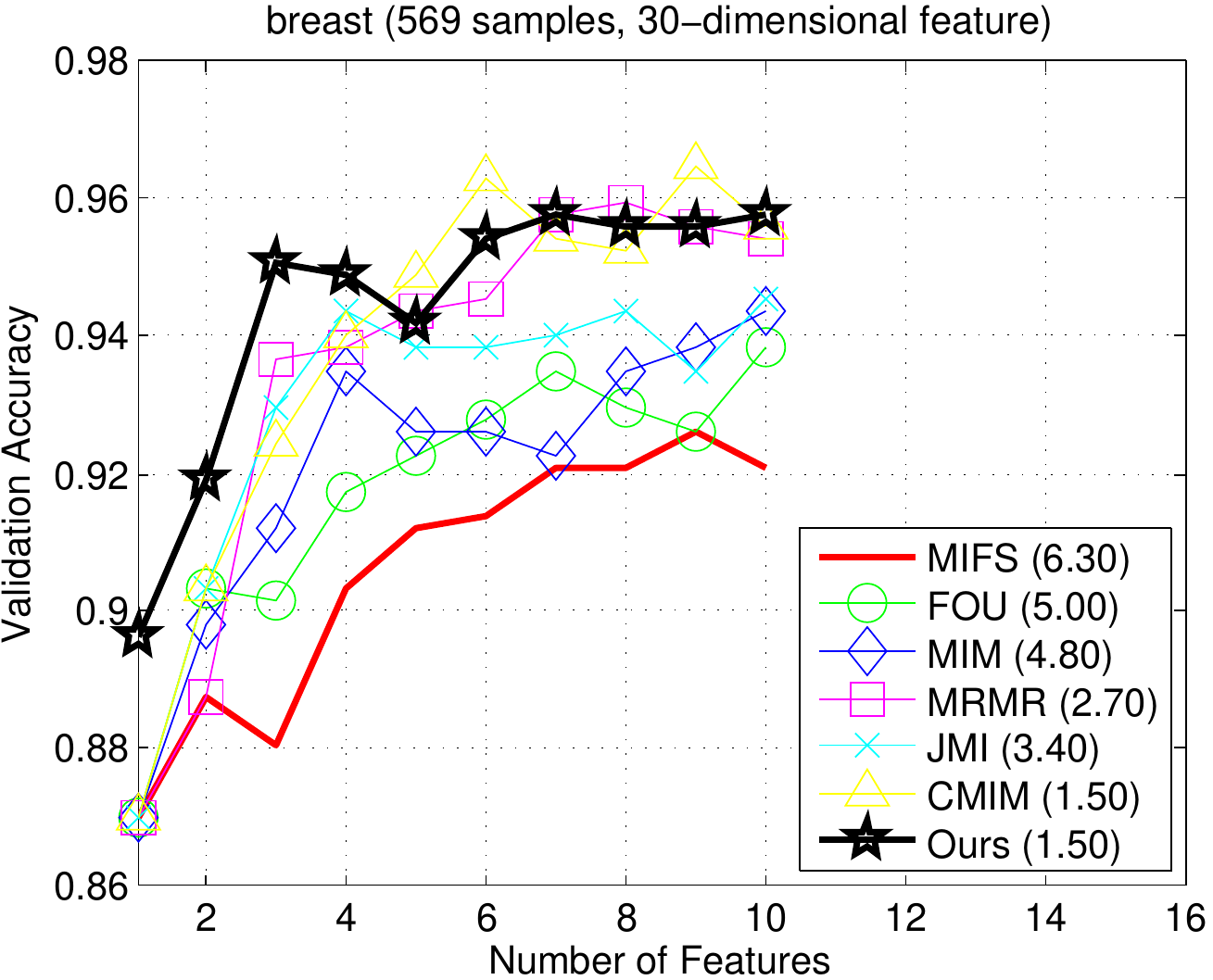}}
\subfigure[semeion] {\includegraphics[width=.32\textwidth]{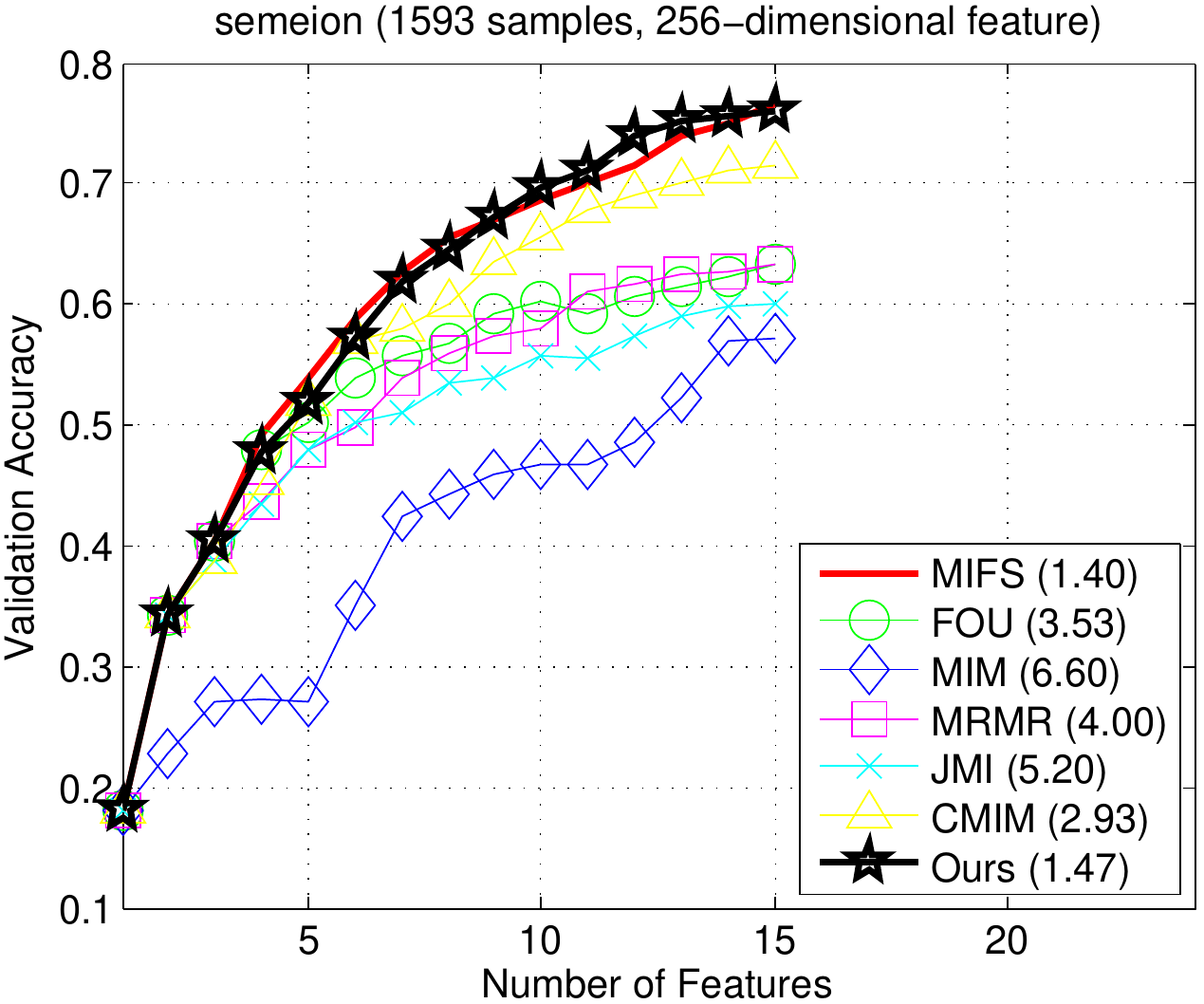}} \\
\subfigure[waveform] {\includegraphics[width=.32\textwidth]{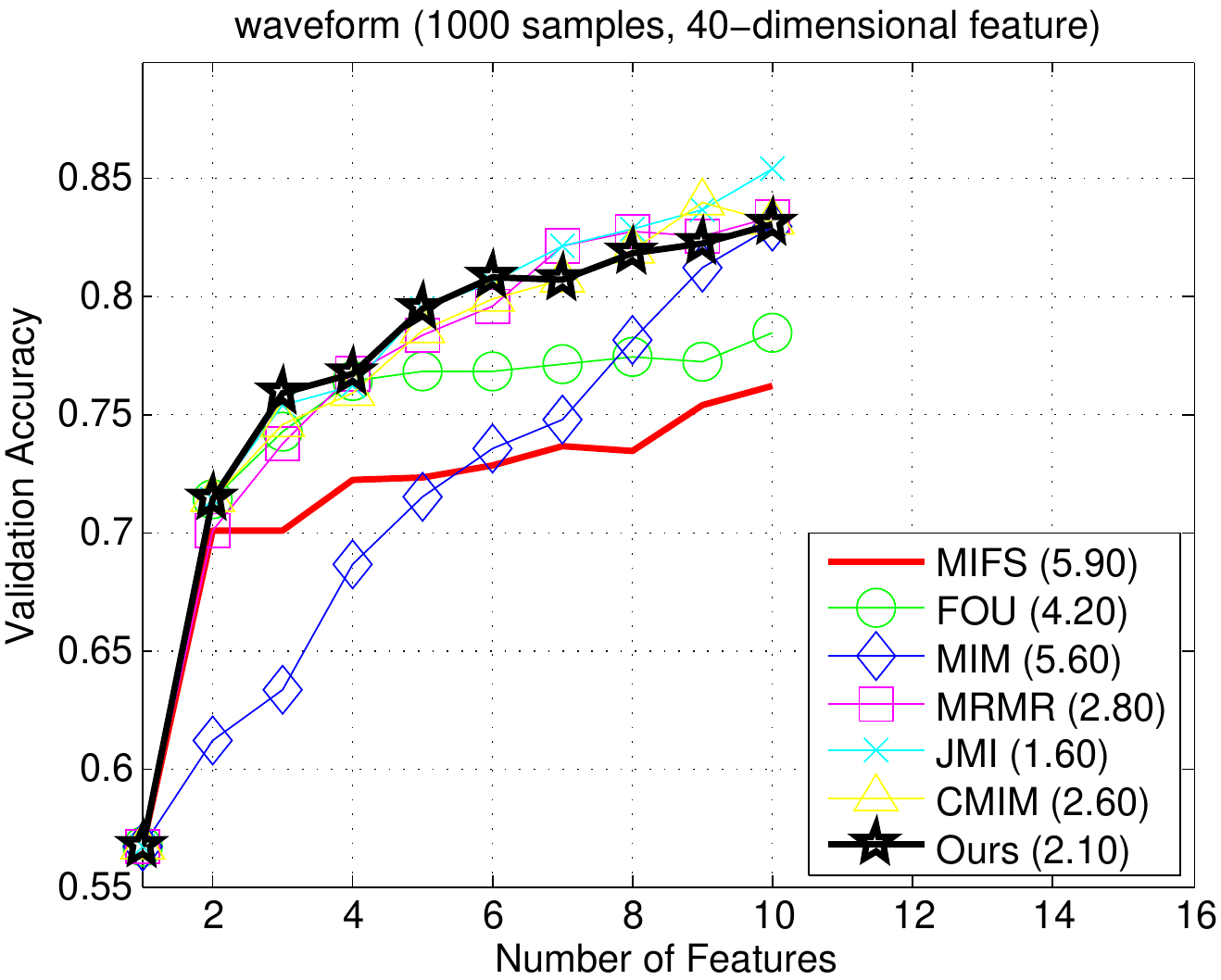}}
\subfigure[Lung] {\includegraphics[width=.32\textwidth]{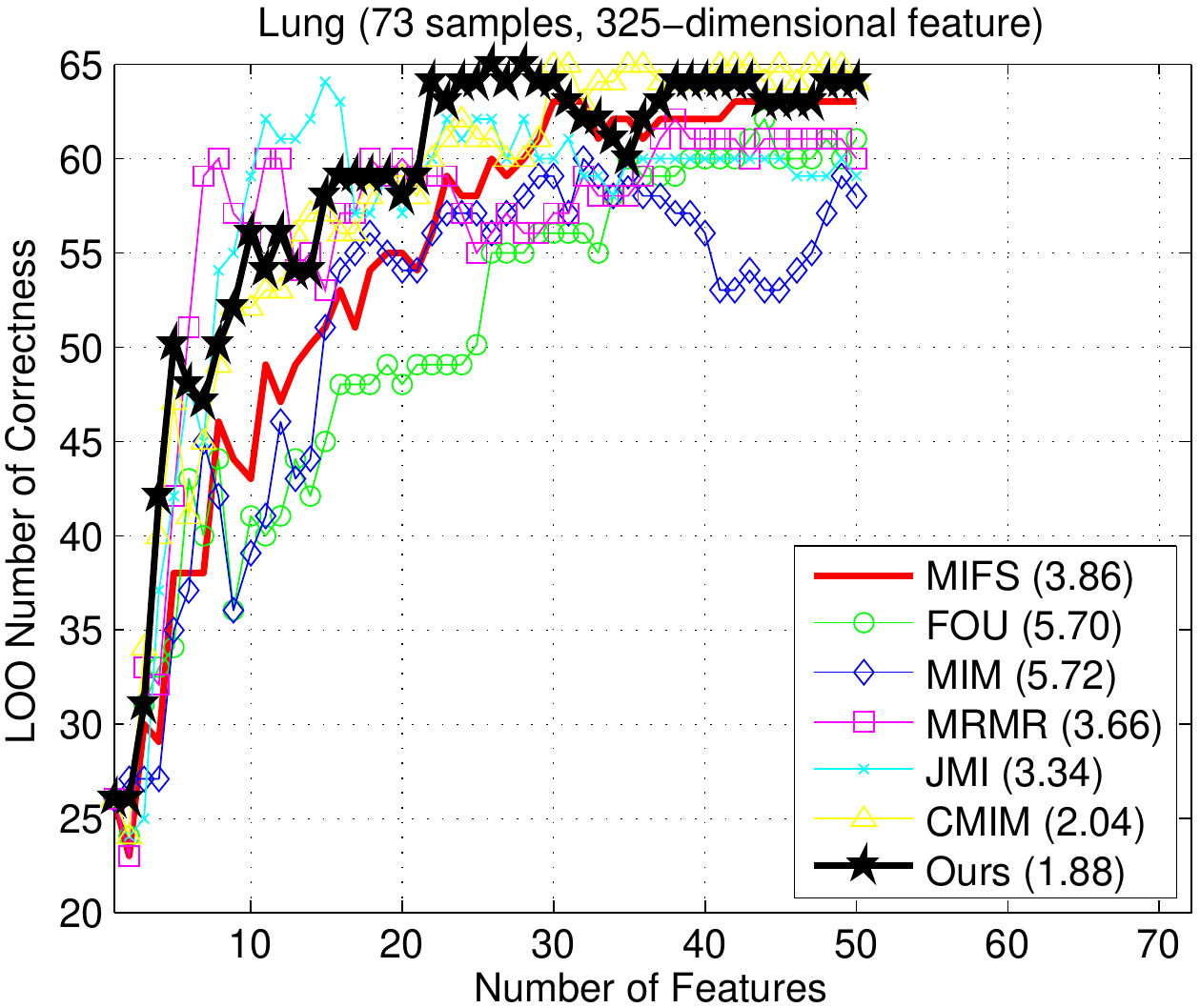}}
\subfigure[Lymph] {\includegraphics[width=.32\textwidth]{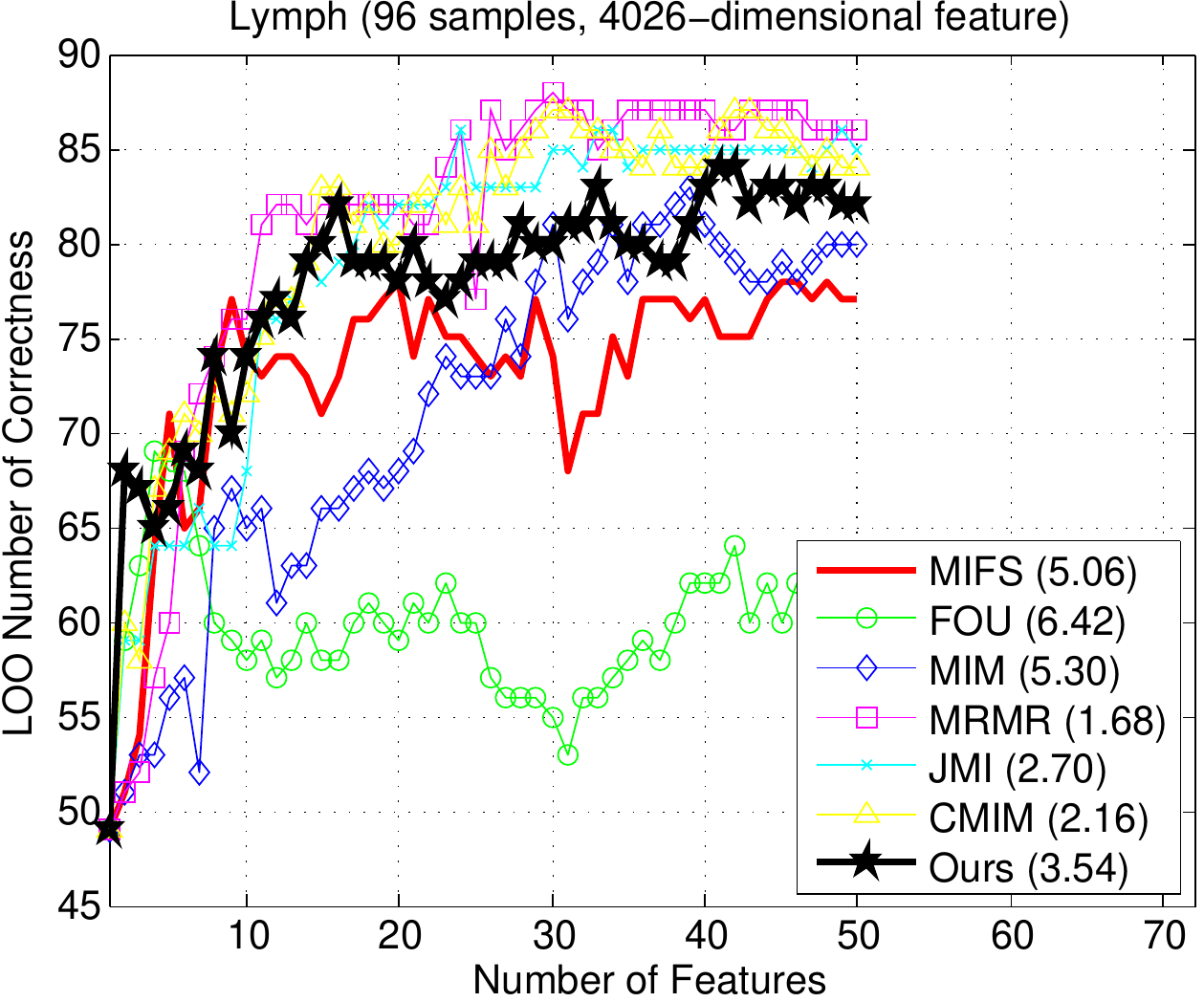}} \\
\subfigure[ORL] {\includegraphics[width=.32\textwidth]{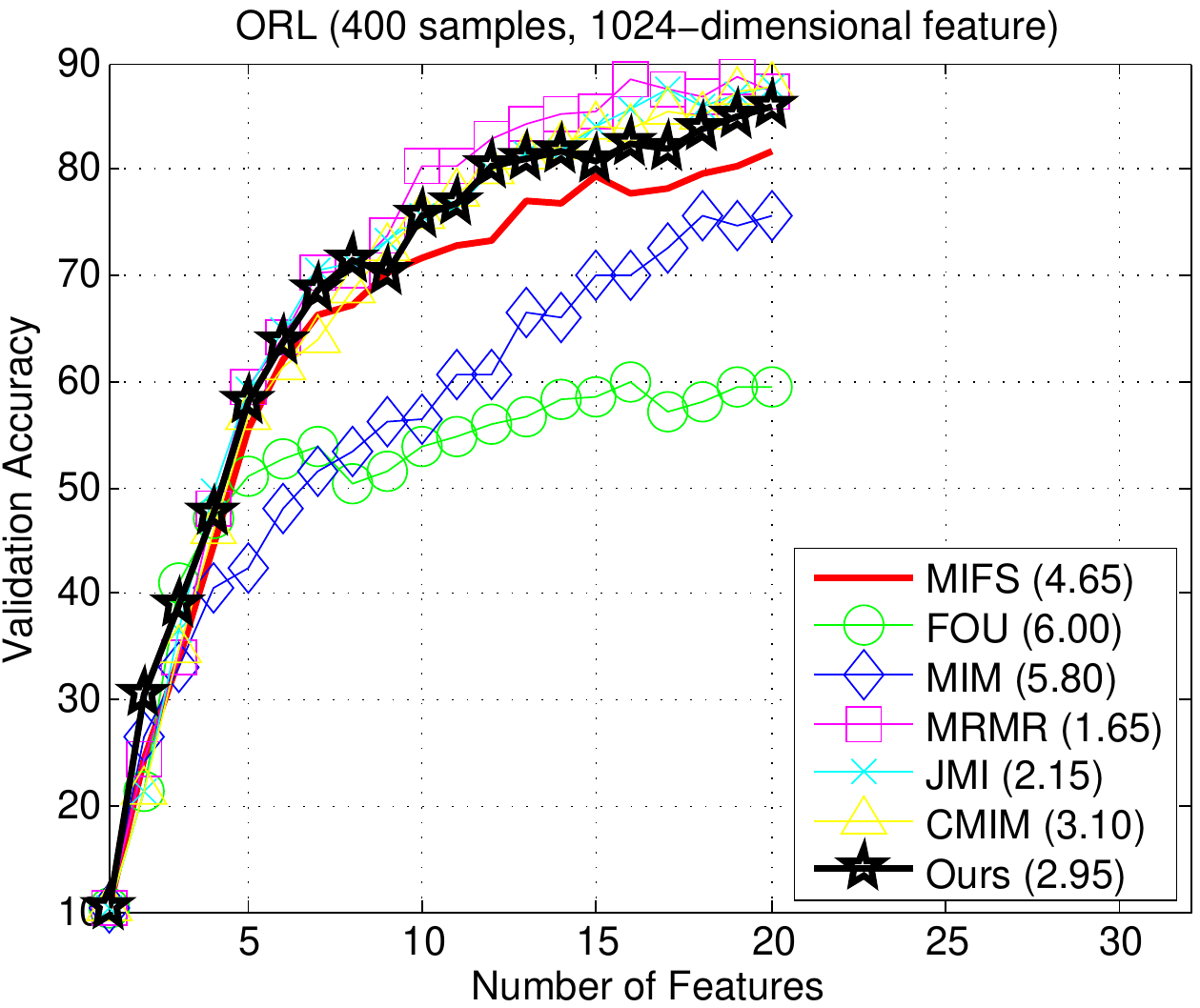}}
\subfigure[warpPIE10P] {\includegraphics[width=.32\textwidth]{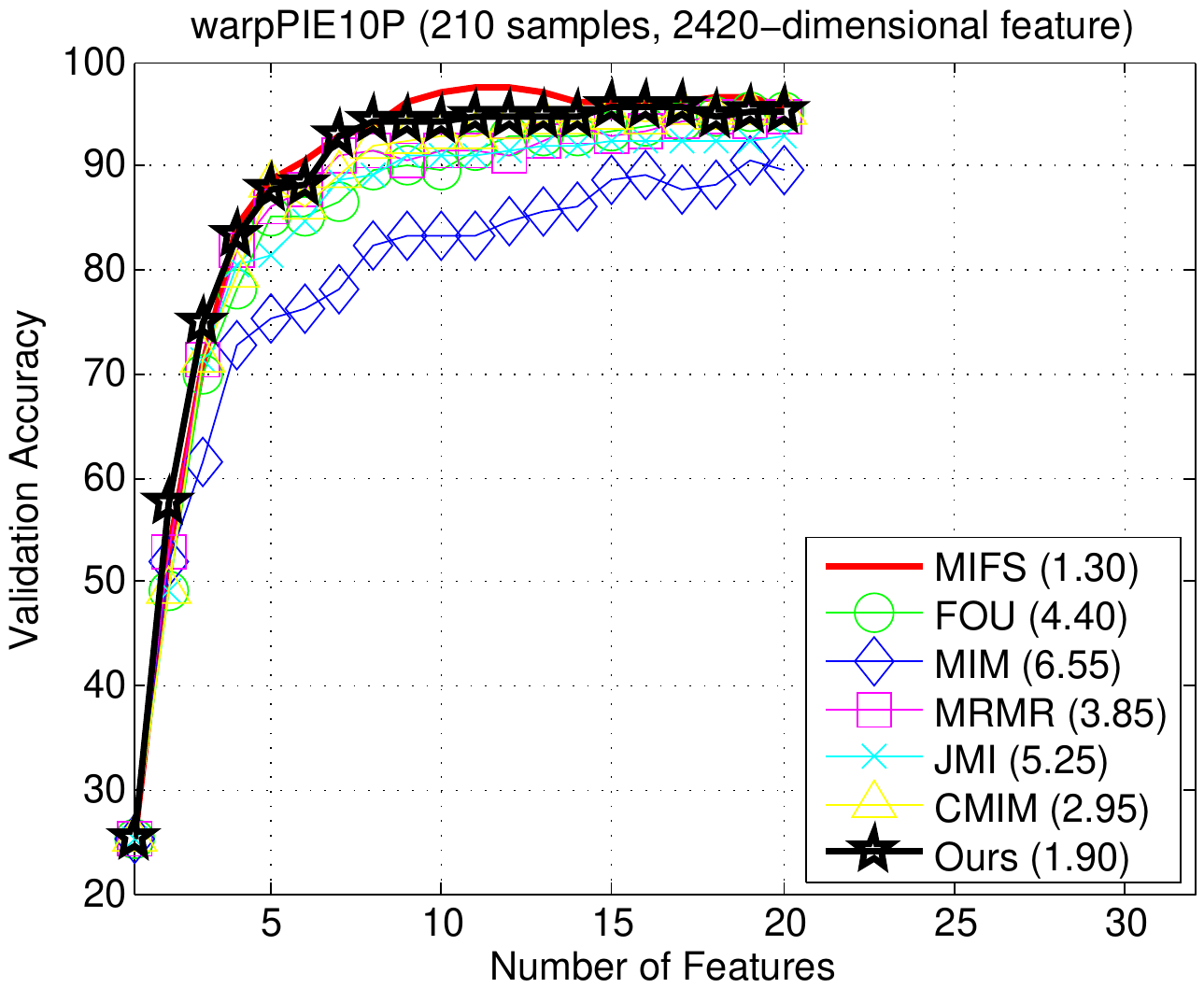}}
\caption{Validation accuracy or Leave-one-out (LOO) results on synthetic and real datasets. The number of samples and the feature dimensionality for each dataset are listed in the title. The value beside each method in the legend indicates the average rank in that dataset. Our method performs favorably in all datasets regardless of data characteristics.}
\label{fig:real_data}
\end{figure*}

\subsection{Real data} \label{app_real}

\begin{table*}
\scriptsize
\centering
\caption{A summarization of different information-theoretic feature selection methods and their average ranks over different number of features in each dataset. The overall average ranks over different datasets are also reported. The best two performance in each dataset are marked with \textcolor{red}{red} and \textcolor{blue}{blue} respectively.}\label{lab:difference_summarization}
\begin{tabular}{ccccccccccc}\hline
 & Criteria & MADELON & breast & semeion & waveform & Lung & Lymph & ORL & PIE & Ave. \\\hline
$\text{MIFS}$~\cite{battiti1994using} & $\mathbf{I}(X_{i_k};Y)-\beta\sum\limits_{l=1}^{k-1}\mathbf{I}(X_{i_k};X_{i_l})$ & $6.65$ & $6.30$ & $\color{red}{1.40}$ & $5.90$ & $3.86$ & $5.06$ & $4.65$ & $\color{red}{1.30}$ & $4.75$ \\
$\text{FOU}$~\cite{brown2009new} & $\mathbf{I}(X_{i_k};Y)-\sum\limits_{l=1}^{k-1}[\mathbf{I}(X_{i_k};X_{i_l})-\mathbf{I}(X_{i_k};X_{i_l}|Y)]$ & $\color{red}{1.90}$ & $5.00$ & $3.53$ & $4.20$ & $5.70$ & $6.42$ & $6.00$ & $4.40$ & $5.13$ \\
$\text{MIM}$~\cite{lewis1992feature} & $\mathbf{I}(X_{i_k};Y)$ & $3.95$ & $4.80$ & $6.60$ & $5.60$ & $5.72$ & $5.30$ & $5.80$ & $6.55$ & $6.13$ \\
$\text{MRMR}$~\cite{peng2005feature} & $\mathbf{I}(X_{i_k};Y)-\frac{1}{k-1}\sum\limits_{l=1}^{k-1}\mathbf{I}(X_{i_k};X_{i_l})$ & $5.85$ & $2.70$ & $4.00$ & $2.80$ & $3.66$ & $\color{red}{1.68}$ & $\color{red}{1.65}$ & $3.85$ & $3.50$ \\
$\text{JMI}$~\cite{yang2000data} & $\sum\limits_{l=1}^{k-1}\mathbf{I}(\{X_{i_k},X_{i_l}\};Y)$ & $2.50$ & $3.40$ & $5.20$ & $\color{red}{1.60}$ & $3.34$ & $2.70$ & $\color{blue}{2.15}$ & $5.25$ & $3.50$ \\
$\text{CMIM}$~\cite{fleuret2004fast} & $\min\limits_l\mathbf{I}(X_{i_k};Y|X_{i_l})$ & $3.60$ & $\color{red}{1.50}$ & $2.93$ & $2.60$ & $\color{blue}{2.04}$ & $\color{blue}{2.16}$ & $3.10$ & $2.95$ & $\color{blue}{2.81}$ \\
$\text{Ours}$ & $\mathbf{I}(\{X_{i_1},X_{i_2},\cdots,X_{i_k}\};Y)$ & $\color{blue}{2.45}$ & $\color{red}{1.50}$ & $\color{blue}{1.47}$ & $\color{blue}{2.10}$ & $\color{red}{1.88}$ & $3.54$ & $2.95$ & $\color{blue}{1.90}$ & $\color{red}{2.19}$ \\\hline
\end{tabular}
\end{table*}

We then evaluate the performance of all methods on $7$ well-known public datasets used in previous research~\cite{brown2012conditional,vinh2014reconsidering}, covering a wide variety of example-feature ratios, class numbers, and different domains including microarray data, image data, biological data, and telecommunication data. Datasets from diverse domains with different characteristics serve as high-quality test bed for a comprehensive evaluation. Different from other datasets, the datasets \emph{Lung} and \emph{Lymph} are already discretized by Peng~\emph{et al.}~\cite{peng2005feature} such that the raw data is not available. This is not a problem for previous information theoretic feature selection methods built upon Shannon's definition. However, our mutual information estimation relies on the Gram matrix evaluated on pairwise samples, this discretization will hurt the ability of our method to take advantage of continuous random variable information, and create an artificial upper limit for performance.

The features within each dataset have a variety of characteristics - some binary/discrete, and some continuous. Following~\cite{brown2012conditional,vinh2014reconsidering}, the base classifier for all data sets is chosen as a linear Support Vector Machine (SVM) (with the regularization parameter set to $1$). The validation results for all competing methods are presented in Fig.~\ref{fig:real_data}(b)-Fig.~\ref{fig:real_data}(h). We also report the ranks in each dataset and the average ranks across all datasets in Table~\ref{lab:difference_summarization}. For each method, its rank in each dataset is summarized as the mean value of ranks across different number of features.

As can be seen, our method can always achieve superior performance on most datasets no matter the number of features. An interesting observation comes from the dataset \emph{breast}, in which our advantage starts from the first feature. This suggests that data discretization will deteriorate mutual information estimation performance, otherwise all the methods will have the same classification accuracy in the first feature (see results on other datasets).

However, the performance of our method is degraded on \emph{Lymph}, as expected. Apart from the improper data discretization (we cannot precisely estimate the Gram matrix because the raw data is unavailable), another reason that causes the degradation is the decreased resolution of our information quantity estimator resulting from a small Gram matrix evaluated on small-sample and high-dimensionality datasets. In fact, it has been observed that our method computes the same value of $\mathbf{I}(\{X_{i_1},X_{i_2},\cdots,X_{i_k}\};Y)$ if $X_{i_k}$ comes from two different feature sources. So a better selection of $\sigma$ should be pursued. 


To complement ranks reported in Table~\ref{lab:difference_summarization}, we perform a Nemenyi's post-hoc test~\cite{NemenyiDistribution} to discover the statistical difference in all competing methods. Specifically, we use the critical difference (CD)~\cite{demvsar2006statistical} as a reference, methods with ranks differ by less than CD are not statistically different, and can be grouped together. The test results are shown in Fig.~\ref{fig:nemenyi_test}, where the black line represents the axis on which the average ranks of methods are drawn, with those appearing on the left hand side performing better. The groups of methods that were not significantly different were connected with a green dashed line. On the one hand, different criteria all achieved visually remarkable improvements against the baseline method MIM (as suggested by the first grouping). On the other hand, only our method and CMIM are significantly different from the baseline method MIM (as suggested by the second grouping).

\begin{figure}[!htbp]
\setlength{\belowcaptionskip}{-0.0cm}
\centering
\includegraphics[width=.70\textwidth]{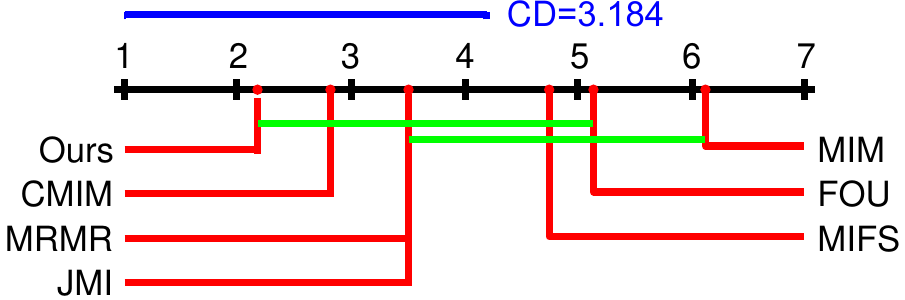}
\caption{Critical difference (CD)~\cite{demvsar2006statistical} generated using Nemenyi's post-hoc test~\cite{NemenyiDistribution} with significance level $0.05$. The groups are identified using the mean rank of a model $\pm$ the CD (marked with a horizontal blue line). There is no evidence of significant differences for models in the same group (joined by the dashed green lines). Our method is the only one that is significantly different from the baseline method MIM.}
\label{fig:nemenyi_test}
\end{figure}

We also analyze the sensitivity to parameters. Our method has two important parameters: the kernel size $\sigma$ and the entropy order $\alpha$. The parameter $\sigma$ controls the locality of our estimator. Theoretically, for small $\sigma$, the Gram matrix approaches identity and thus its eigenvalues become more similar, with $1/n$ as the limit case. By contrast, for large $\sigma$, the Gram matrix approaches all-ones matrix as the limit case and its eigenvalues become zero except for the one. Therefore, extremely small and large values of $\sigma$ are of limited interests. We expect our estimator to work well in a large range of $\sigma$, because this application concerns more on the large/small relationships among several measurements (rather than their specific values), and these relationships will not be affected if the value of $\sigma$ does not result in the saturation of mutual information estimation. However, a relatively large $\sigma$ (in a reasonable range) is still preferred. This is because both entropy and mutual information monotonically increase as $\sigma$ decreases~\cite{giraldo2015measures}, large $\sigma$ makes the mutual information between labels and selected feature subset increase slowly, thus encouraging the discriminability if we are going to continue the selection. We investigate how $\sigma$ affects the performance of our method for different values, $\{0.1,0.5,1,5,10,50,100\}$. We also evaluate our performance with $\sigma$ tuned with $10$ to $20$ percent of the total (median) range of the Euclidean distances between all pairwise data points~\cite{shi2000normalized}. For example, in dataset \emph{semeion}, this range corresponds to $9.17<\sigma<9.64$. Performance variance result is presented in Fig.~\ref{fig:parameter_robustness}. Due to space limitations, we only report the results in terms of validation accuracy on \emph{waveform} and \emph{semeion} datasets. We can observe that the accuracy values and the average ranks are not sensitive to $\sigma$ in a large range ($0.5$ to $10$ in \emph{waveform}, $0.1$ to $10$ in \emph{semeion}), but relatively large $\sigma$ seems better.

As discussed earlier, $\alpha$ changes the emphasis from the tails of the distribution (smaller $\alpha$) to places with large concentration of mass (larger $\alpha$)~\cite{principe2010information}. Since classification uses a counting norm, values lower than $2$ are preferred. We use $\alpha=1.01$ to approximate Shannon's entropy and make the comparison fair. We also observed a performance gain when $\alpha$ is smaller than $1$. See Appendix~\ref{appendix_A} and~\ref{appendix_B} for results.


\begin{figure}[!t]
\setlength{\abovecaptionskip}{0.cm}
\setlength{\belowcaptionskip}{-0.0cm}
\centering
\subfigure[waveform] {\includegraphics[width=.45\textwidth]{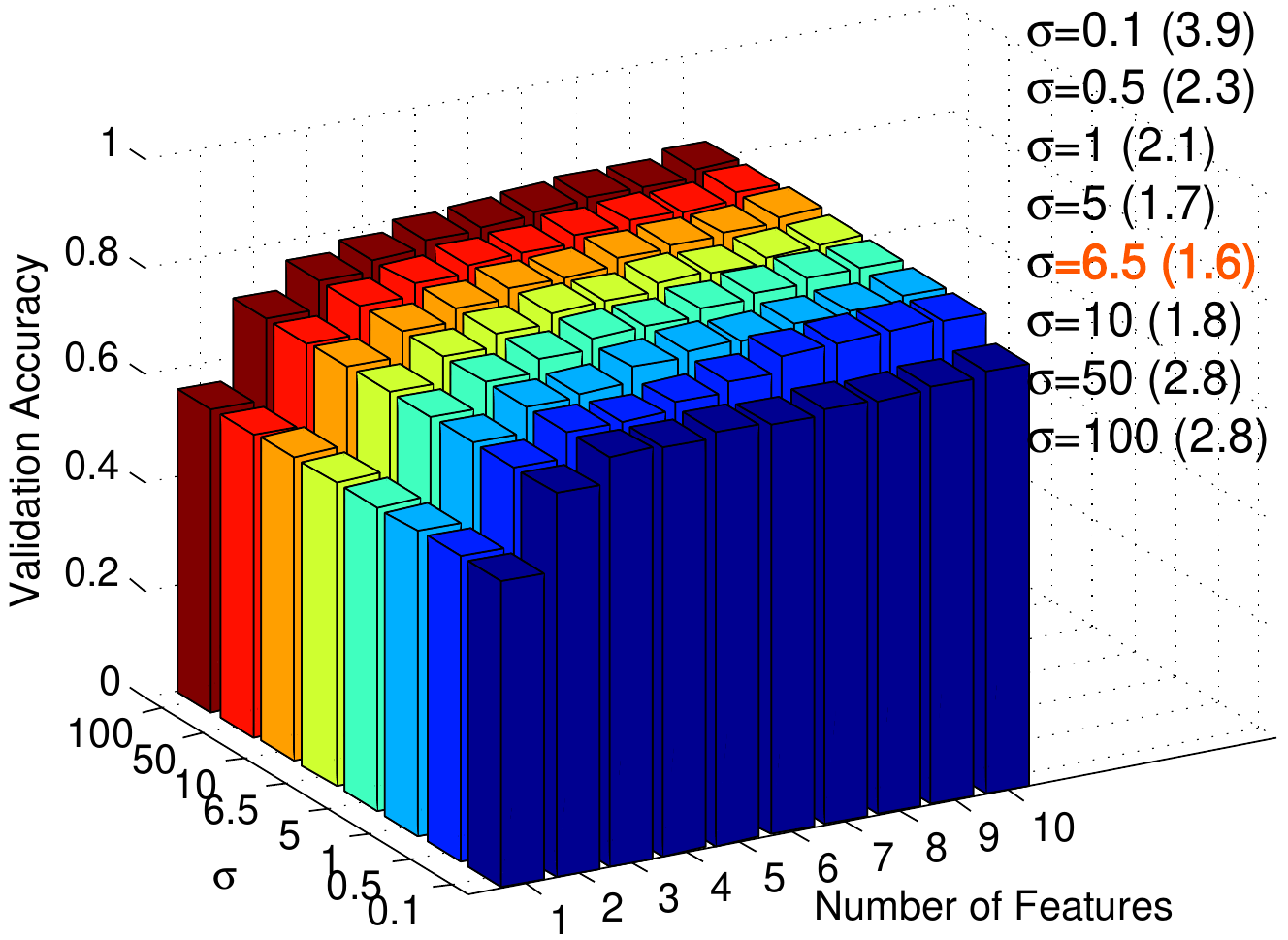}}
\subfigure[semeion] {\includegraphics[width=.45\textwidth]{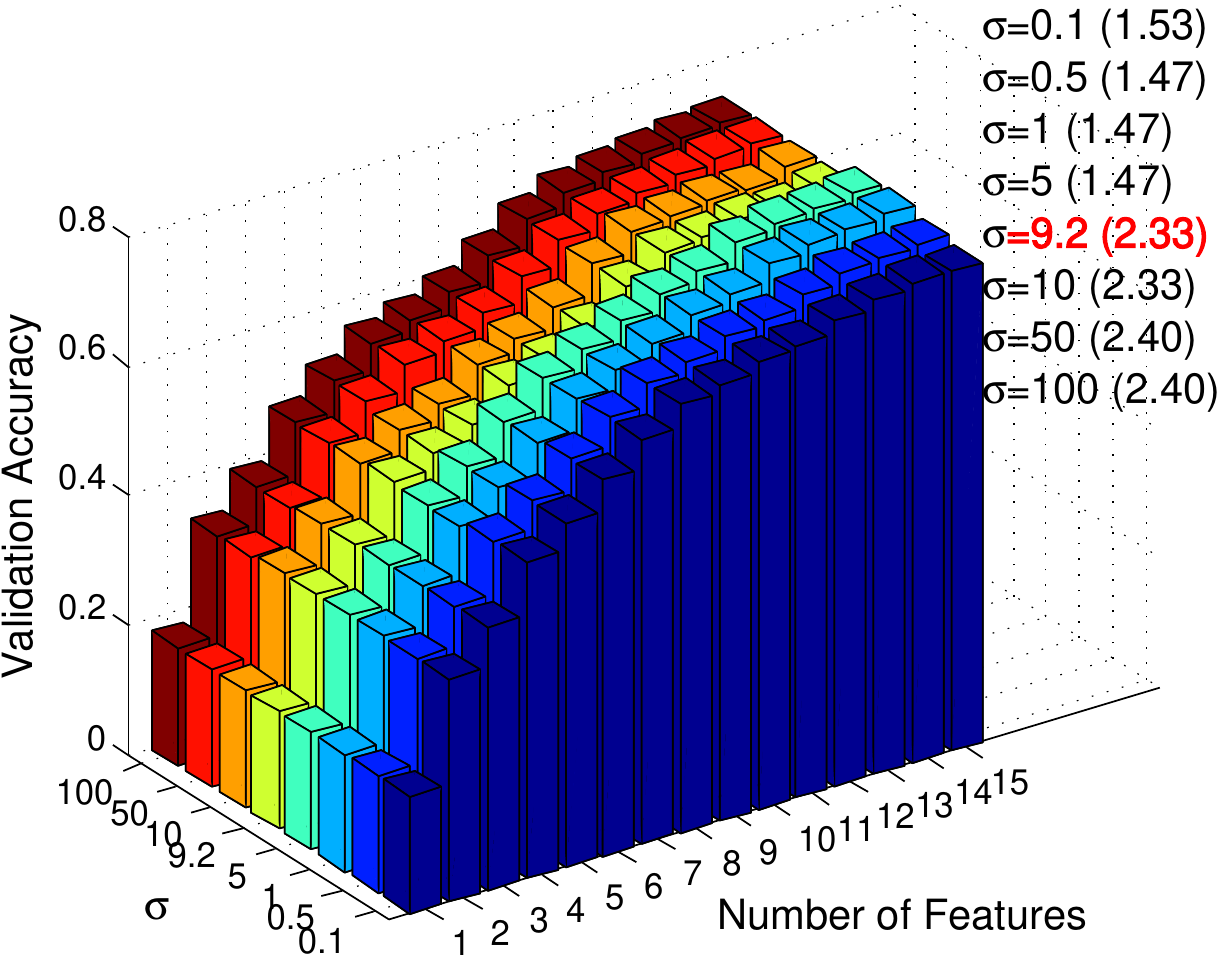}} \\
\caption{Validation accuracy of our method on waveform and semeion datasets with respect to different $\sigma$ and feature numbers ($\alpha=1.01$). The value in the parenthesis is the average rank of our method with respect to other competing methods, when using the corresponding $\sigma$. Our method works well in a large range of $\sigma$, and the $10$ to $20$ percent of the total (median) range is a more reliable choice to $\sigma$ (marked with red).}
\label{fig:parameter_robustness}
\end{figure}

We finally briefly analyze the computational and memory cost of different methods. Let $d$ be the number of features and $n$ the number of samples. Apart from MIM that only requires one sort, all methods use the same forward selection scheme and require $d^2/2$ times sort. The main difference in computational complexity comes from the estimation of mutual information. Our estimator deals with continuous (or mixed) random variables, it takes $\mathcal{O}(n^3)$ time for the eigenvalue decomposition of a $n\times n$ Gram matrix. Other methods focus on discrete random variables, which takes roughly $\mathcal{O}(n)$ time. However, if they substitute Shannon's discrete entropy with the differential entropy, the continuous PDF estimation typically take $\mathcal{O}(n^2)$ time~\cite{principe2010information}. As for the memory cost, MIM just needs to reserve $d$ mutual information values, our method needs to reserve $(d+1)$ Gram matrices of size $n\times n$, whereas others need to reserve all pairwise mutual information values. A summary is given in Table~\ref{lab:cost_summarization}. Admittedly, the computational complexity is higher for the original formulation of the matrix-based quantities. It is possible to apply methods such as kernel randomization~\cite{lopez2013randomized} to reduce the burden to $\mathcal{O}(n\log(n))$. Please also note that, the differentiability of the matrix-based objective opens the door to other search techniques beyond the greedy selection. We leave this to future work.

\begin{table}
\setlength{\abovecaptionskip}{-0.0cm}
\setlength{\belowcaptionskip}{-0.0cm}
\centering
\caption{Summary of computational and memory cost.}\label{lab:cost_summarization}
\begin{tabular}{ccccc}\hline
 & Computational & Memory \\\hline
$\text{MIM}$ & $\mathcal{O}(nd)$ & $\mathcal{O}(d)$ \\
$\text{Ours~(continuous)}$ & $\mathcal{O}(n^3d^2)$ & $\mathcal{O}(n^2d)$ \\
$\text{Others~(discrete)}$ & $\mathcal{O}(nd^2)$ & $\mathcal{O}(d^2)$ \\
$\text{Others~(continuous)}$ & $\mathcal{O}(n^2d^2)$ & $\mathcal{O}(d^2)$ \\\hline
\end{tabular}
\end{table}

\subsection{Feature selection for hyperspectral image (HSI) classification}
We finally evaluate the performances of all methods in another real example of great importance: band selection for hyperspectral image (HSI) classification. In particular, the spectrum of each pixel (consisting of measurements using hundreds of spectral bands) is a very popular feature in the literature. However, this kind of data is usually noisy and contains high redundancy between adjacent bands~\cite{feng2016multiple}. Therefore, it would be very helpful if one could select a subset of spectral bands (i.e., the most important bands or wavelengths) beforehand. For some applications, these spectral bands can be used to infer mineralogical and chemical properties~\cite{yu2018band}.

We apply all feature (here referring to bands) selection methods mentioned in sections~\ref{app_artificial} and~\ref{app_real} on the publicly available benchmark Indian Pine data~\cite{landgrebe2005signal}, consisting of $145\times145$ pixels by $220$ bands of reflectance Airborne Visible/Infrared Imaging Spectrometer (AVIRIS). Because of atmospheric water absorption, a total of $20$ bands can be identified as noisy ($104$-$108$, $150$-$163$, $220$) and safely removed as a preprocessing procedure~\cite{camps2005kernel}. There are $10,366$ labeled pixels from $16$ classes such as corn and grass.

We test the performances of all methods on three different gallery (training set) sizes, i.e., for each class, $1\%$, $5\%$ and $10\%$ of the available labeled samples were randomly selected as the gallery. The remaining samples were then used as the probe set for evaluation. For each gallery size, the random selection process was repeated $10$ times, and the average quantitative evaluation metrics among $10$ simulations were recorded. We choose SVM with a RBF kernel as the baseline classifier, as it is the most widely used method in HSI classification research~\cite{camps2014advances}. The overall accuracy (OA) and average accuracy (AA) are adopted as the objective metrics to evaluate HSI classification results. The OA is computed as the percentage of correctly classified test pixels, whereas the AA is the mean of the percentage of correctly classified pixels for each class.

\begin{figure}[!t]
\setlength{\abovecaptionskip}{0.cm}
\setlength{\belowcaptionskip}{-0.0cm}
\centering
\subfigure[OA] {\includegraphics[width=.45\textwidth]{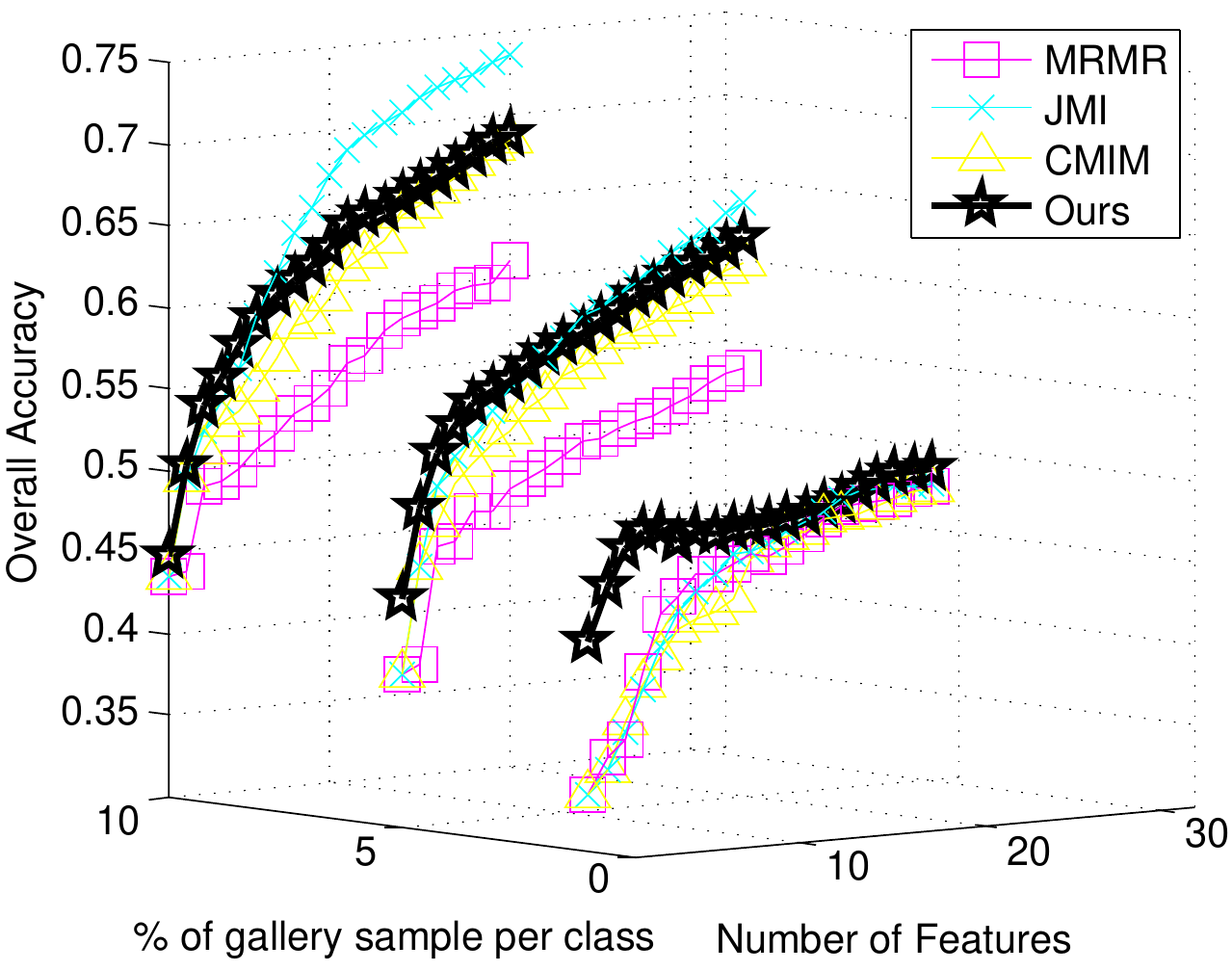}}
\subfigure[AA] {\includegraphics[width=.45\textwidth]{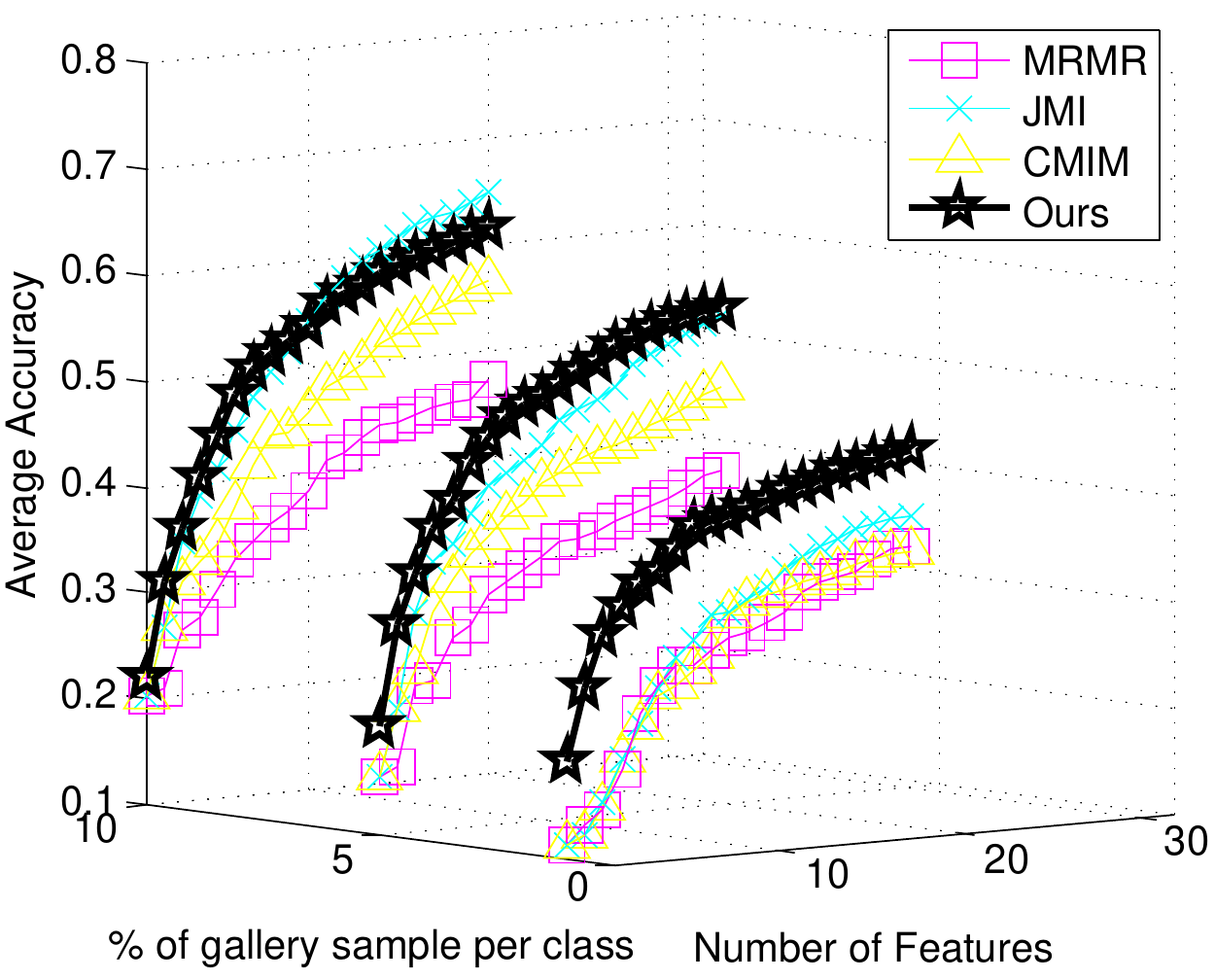}}
\caption{Overall accuracy (OA) and average accuracy (AA) of competing methods in three different gallery sizes ($1\%$, $5\%$ and $10\%$ gallery samples per class) on Indian Pine dataset. MIFS, FOU and MIM perform poorly in this example, and thus omitted. Our method is consistently better than other in the scenario of $1\%$ and $5\%$ gallery sizes, but inferior to JMI using $10\%$ gallery samples per class.}
\label{fig:indian_pine}
\end{figure}

The quantitative validation results are shown in Fig.~\ref{fig:indian_pine}. As can be seen, our method always provides consistently higher OA and AA values when the gallery is small. In the case of $1\%$ gallery samples, our OA values have a fluctuation after selecting $5$ bands. This is probably because the training samples are rather limited ($\sim7$ per class), thus a small perturbation in the classification hyperplane may result in a large change in OA or AA.
JMI outperforms our method given sufficient amount of training data. However, according to Fig.~\ref{fig:spectral_measurement}, the bands selected by JMI are not stable across $10$ runs, which makes JMI poor for interpretability~\cite{li2017feature}. In fact, the bands selected by our method are dispersed, which gives higher opportunity to provide complementary information, since the adjacent bands are rather redundant in HSI~\cite{feng2016multiple}. Meanwhile, these bands covers most regions with the large interval of the reflectance spectrums, indicating its highly discriminative ability for different categories~\cite{feng2016multiple}. Moreover, our method consistently selects bands $1$, $25$, $35$, $57$, $75$ and $89$ regardless of training data perturbations, which is consistent with previous work on band selection from different perspectives~\cite{jia2012unsupervised,yu2018band}, where the bands $1$, $23$-$27$, $30$-$36$, $57$, $75$-$78$ and $87$-$91$ are frequently selected.

\begin{figure}[!t]
\setlength{\abovecaptionskip}{0.cm}
\setlength{\belowcaptionskip}{-0.0cm}
\centering
\subfigure[JMI] {\includegraphics[width=.45\textwidth]{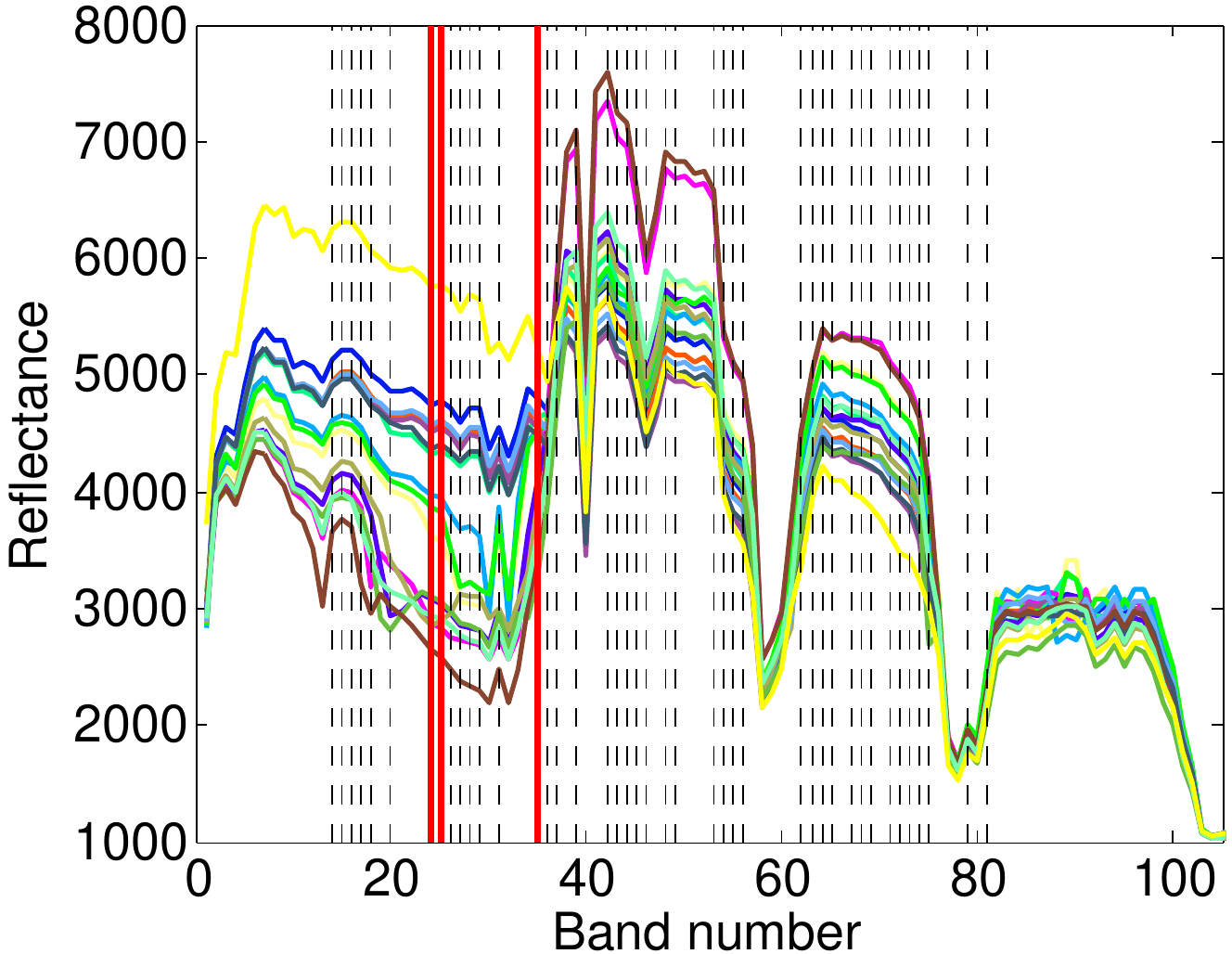}}
\subfigure[our method] {\includegraphics[width=.45\textwidth]{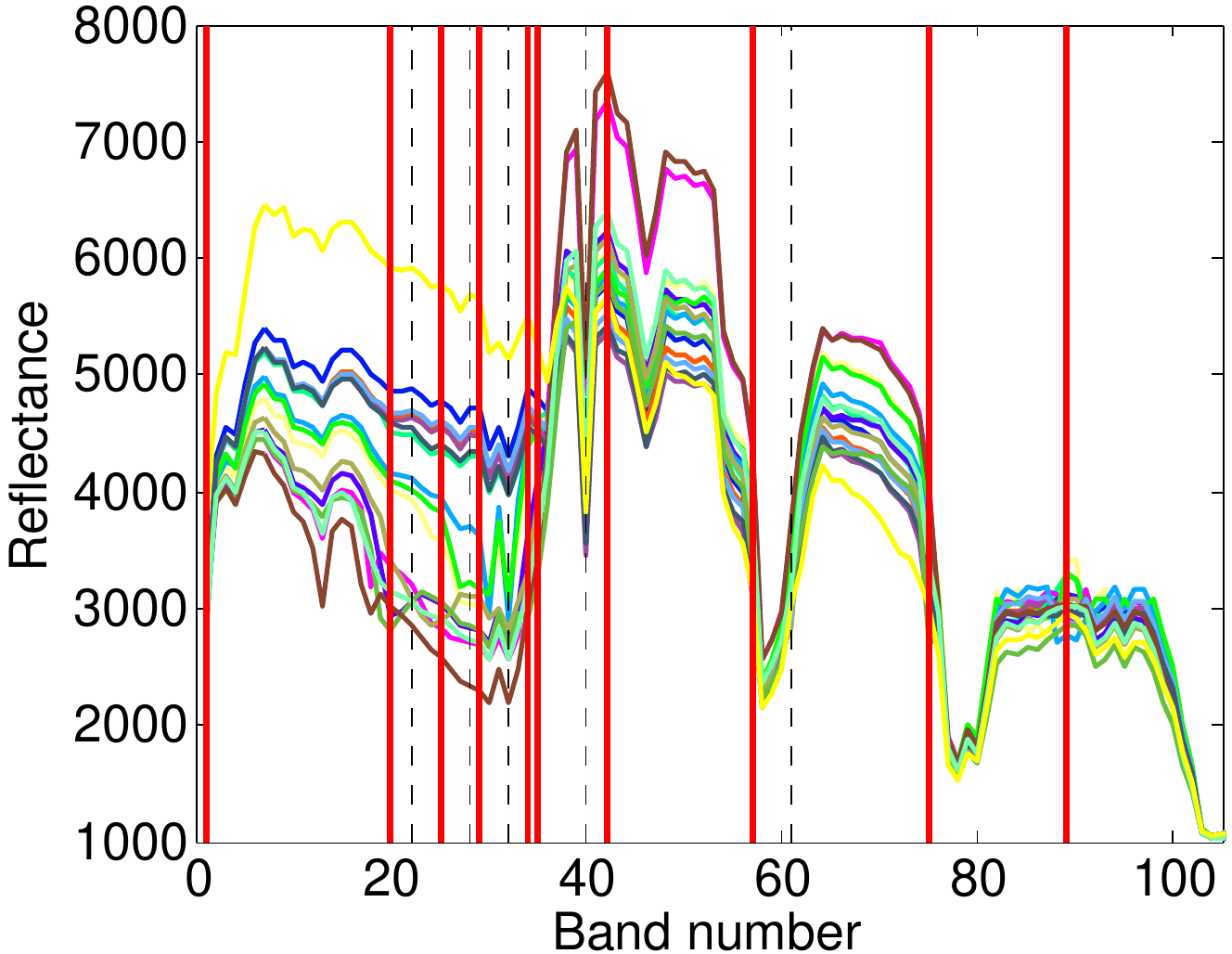}} \\
\subfigure[Legend of different land-cover categories] {\includegraphics[width=.90\textwidth]{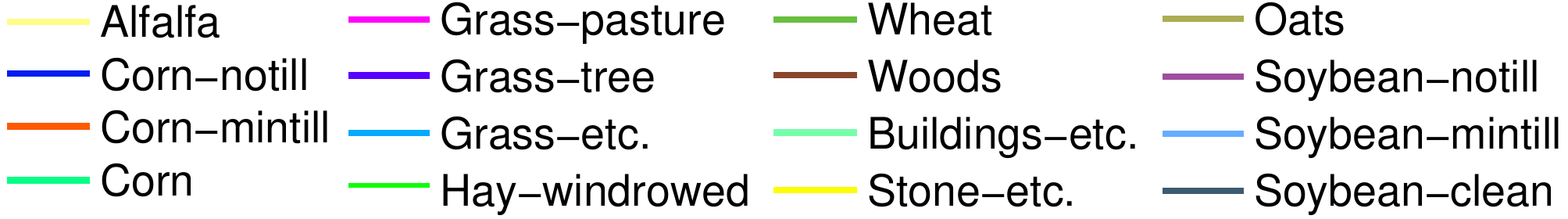}}
\caption{Reflectance of $16$ land-cover categories on the Indian Pine data. Only spectral bands between $1$ and $103$ are displayed, since the $10$ most significant bands selected by JMI and our method all lie in this region. In (a) and (b), the black dashed lines represent the bands selected by JMI and our method respectively over $10$ independent runs, whereas the red solid line highlight the bands that have been selected at least $7$ times in the $10$ runs. (c) shows the land-cover legend. Our method is stable with respect to training data perturbations.}
\label{fig:spectral_measurement}
\end{figure}

Finally, by referring to the classification maps shown in Fig.~\ref{fig:spectral_segmentation}, our method improves the region uniformity~\cite{feng2016multiple} of the grass-pasture, hay-windrowed and soybean-clean (marked with white rectangles) in comparison to JMI, although both methods offer similar OA and AA values.

\begin{figure*}[!t]
\setlength{\abovecaptionskip}{0.cm}
\setlength{\belowcaptionskip}{-0.0cm}
\centering
\subfigure[Ground truth] {\includegraphics[width=.32\textwidth]{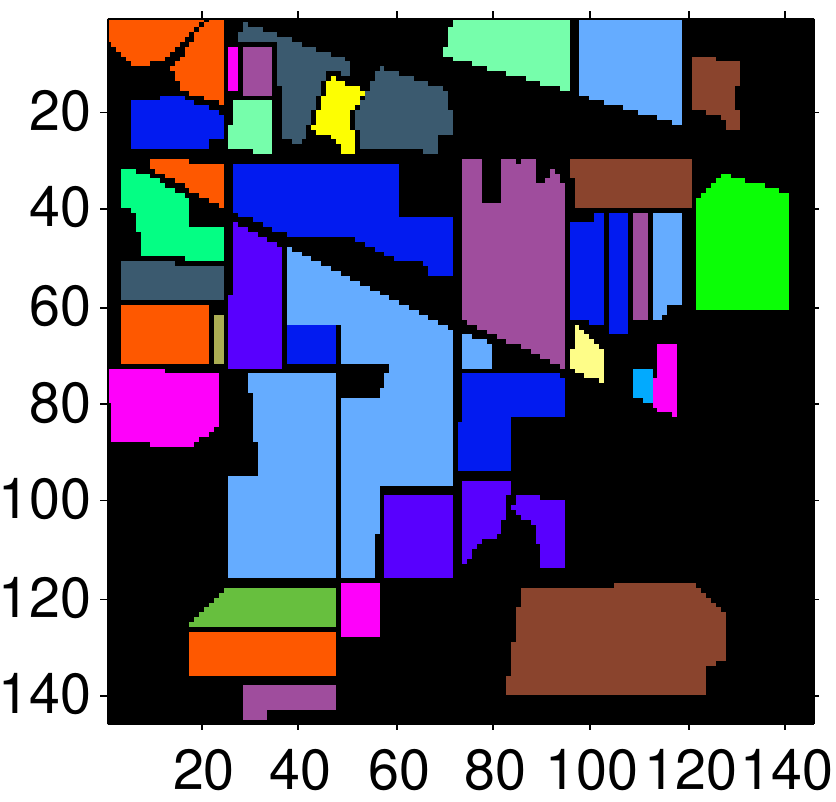}}
\subfigure[JMI] {\includegraphics[width=.32\textwidth]{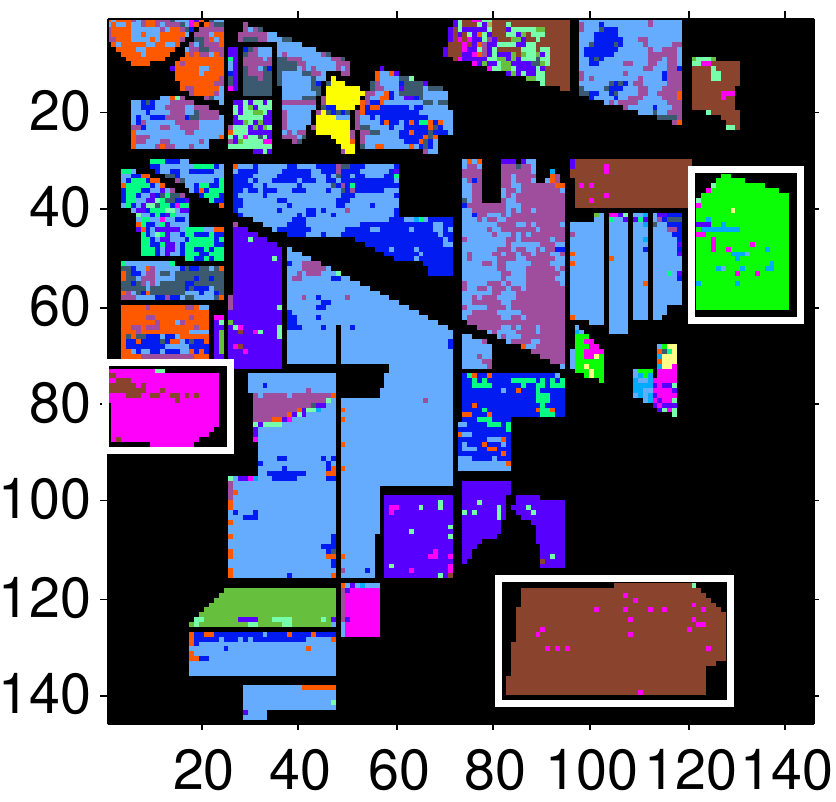}}
\subfigure[our method] {\includegraphics[width=.32\textwidth]{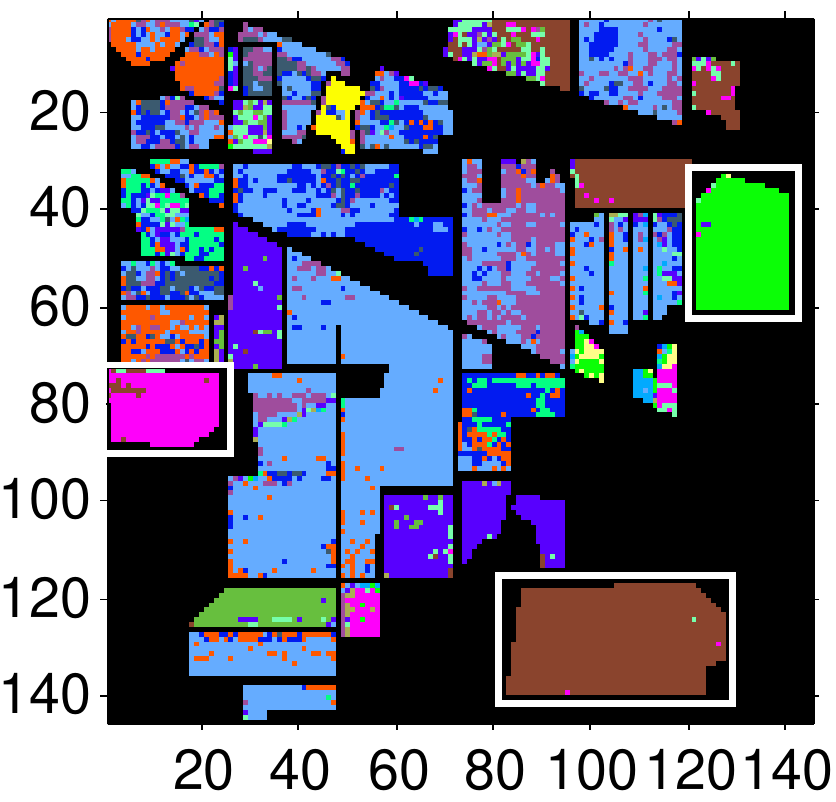}}
\caption{Classification maps on the Indian Pine data. $5\%$ samples from each class care selected to constitute gallery set. The legend is same to Fig.~\ref{fig:spectral_measurement}. Our method improves the uniformity of regions marked with white rectangles.}
\label{fig:spectral_segmentation}
\end{figure*}

\section{Conclusions} \label{conclusions}

In this paper, we generalize the matrix-based R{\'e}nyi's $\alpha$-order joint entropy to multiple variables. The new definition enables us to efficiently and effectively measure various multivariate interaction quantities, such as the interaction information and the total correlation. We finally present a real application on feature/band selection to show how our matrix definition works well, closely matching the ideal mutual information objective without any approximation or decomposition.


In the future, we will explore more machine learning applications in more complex scenarios involving high-dimensional data and complex dependence structure, such as understanding the learning dynamics of deep neural network (DNNs) with information theoretic concepts~\cite{yu2018understanding}. At the same time, we will investigate novel information theoretic objectives to further improve feature selection performance. One possible solution is to precisely determine the redundancy and synergy among different features using the partial information decomposition (PID) framework~\cite{williams2010nonnegative}.

\section*{Acknowledgment}
This work was funded in part by the U.S. ONR under grant N00014-18-1-2306, in part by the DARPA under grant FA9453-18-1-0039, and in part by the Norwegian Research Council FRIPRO grant no. 239844 on developing the Next Generation Learning Machines.

%

\bibliographystyle{IEEEabrv}
\bibliography{ITL_multivariate}

\appendix
\subsection{Feature selection results with $\alpha=0.6$}\label{appendix_A}

We first present feature selection results with $\alpha=0.6$. The quantitative evaluation results are shown in Fig.~\ref{fig:real_data_0_6} and Table~\ref{lab:difference_summarization_0_6}. As can be seen, the advantage of our method becomes more obvious, even superior to the case of $\alpha=1.01$. One possible reason is that the closer $\alpha$ is to zero, the more relevance of our estimator is put on the tails of the distribution, hence the better the method will be in the boundaries between classes (since classification uses a counting norm). However, practitioners should also be aware that the estimation with small $\alpha$ requires more samples to get a better consistency of the tails, which do not always occur. These results further confirm our argument that values of $\alpha$ lower than $2$ are preferred in the application of feature selection.

\begin{figure*}[!t]
\setlength{\abovecaptionskip}{0.cm}
\setlength{\belowcaptionskip}{-0.0cm}
\centering
\subfigure[MADELON] {\includegraphics[width=.23\textwidth,height=3cm]{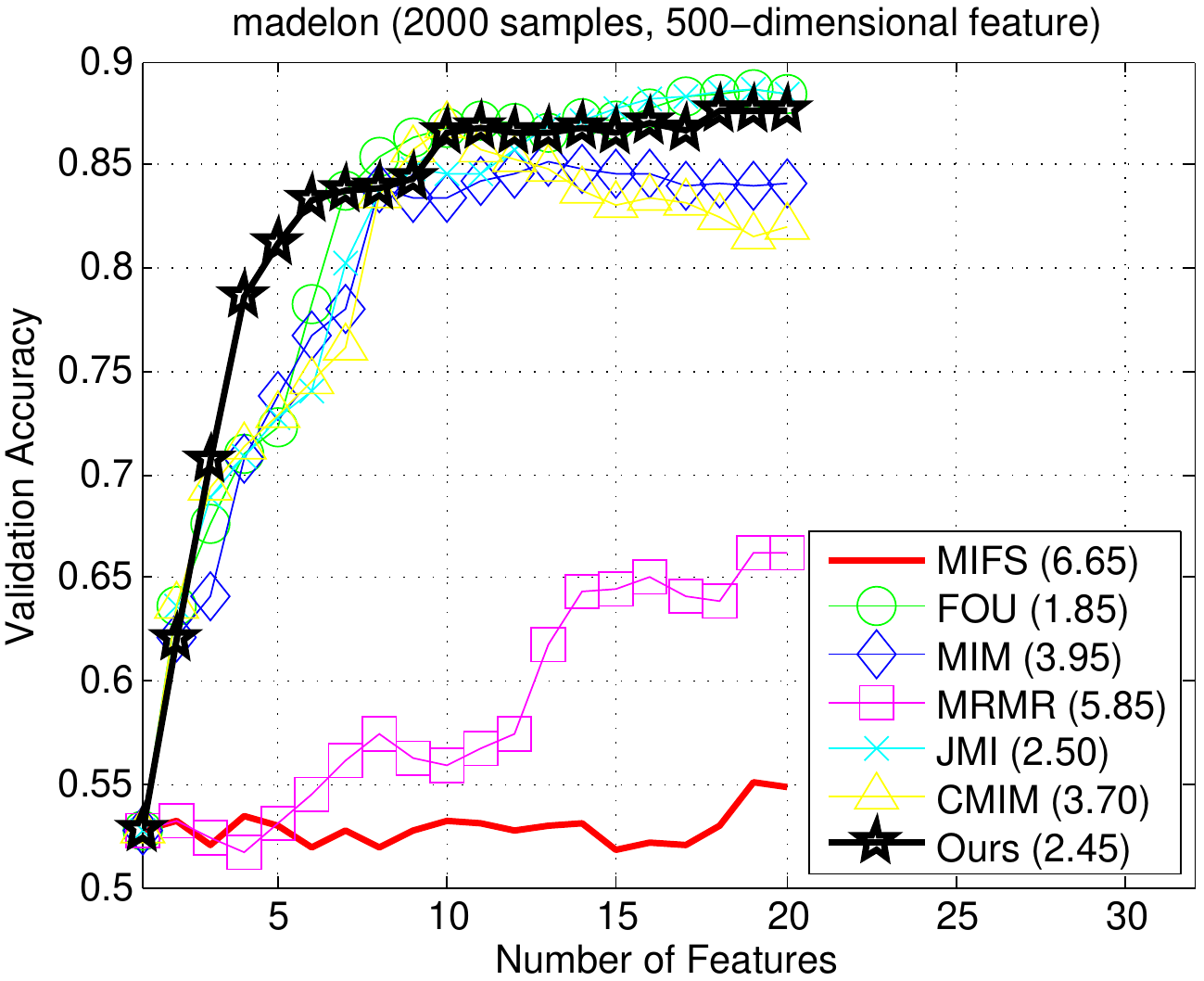}}
\subfigure[breast] {\includegraphics[width=.23\textwidth,height=3cm]{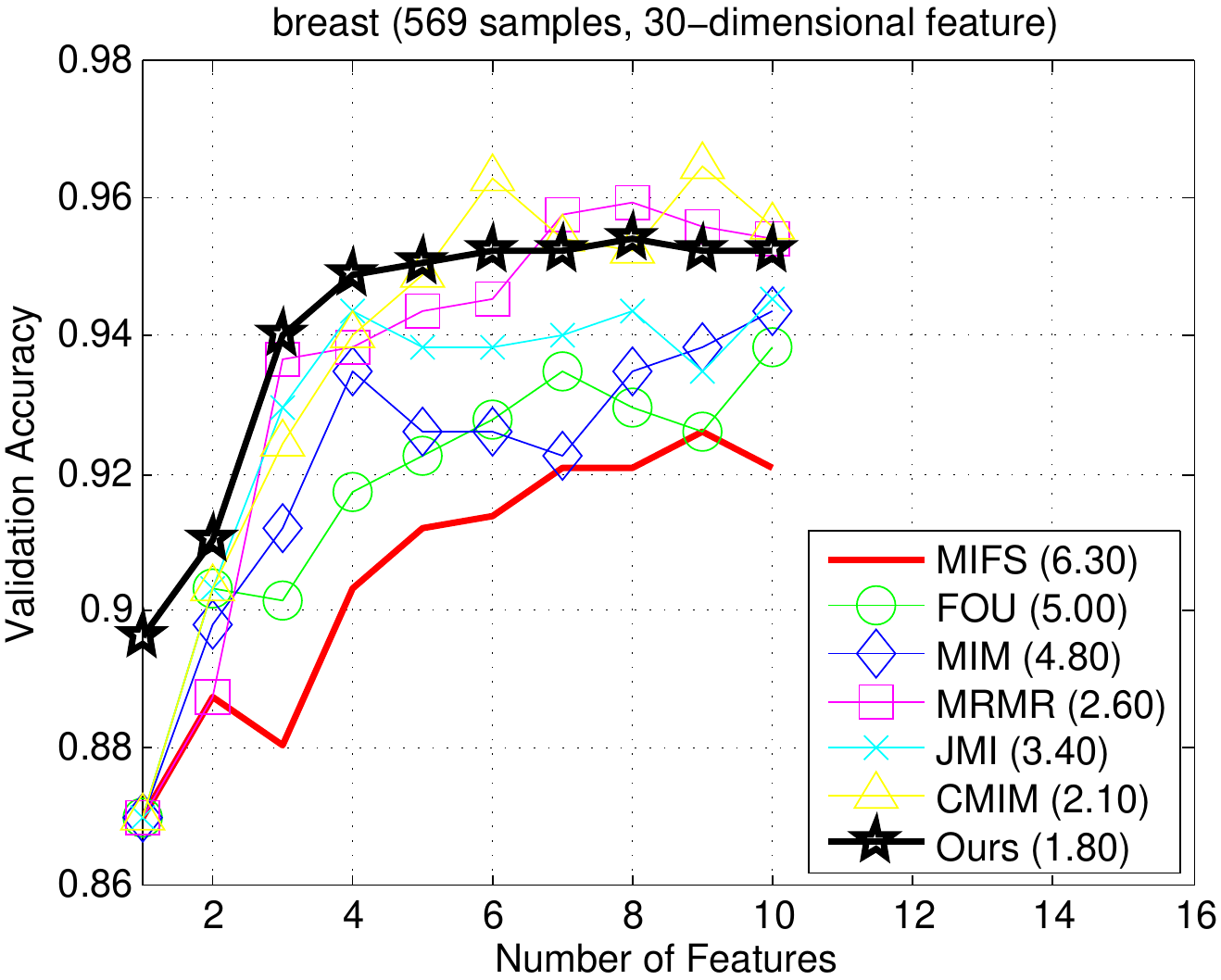}}
\subfigure[semeion] {\includegraphics[width=.23\textwidth,height=3cm]{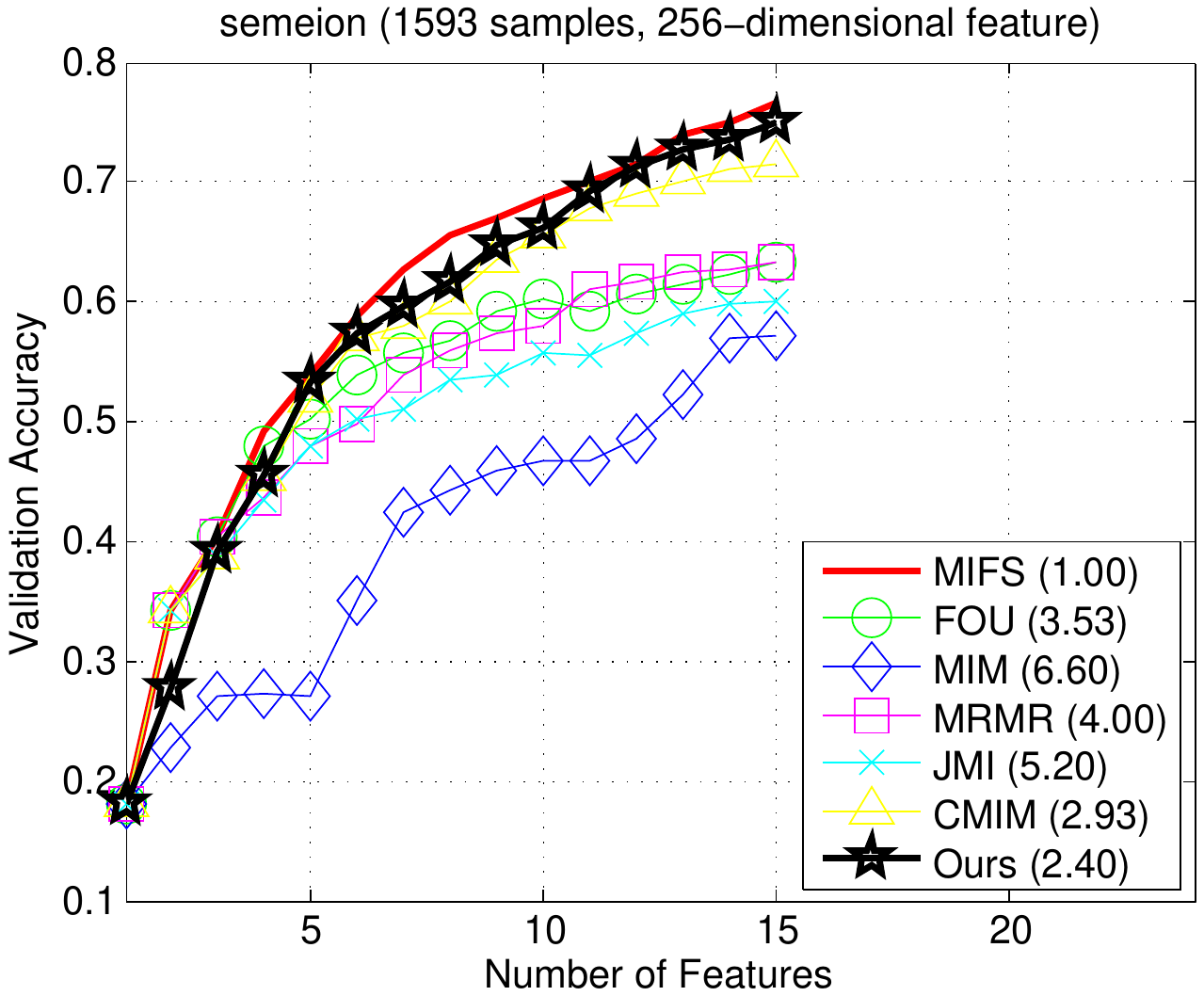}}
\subfigure[waveform] {\includegraphics[width=.23\textwidth,height=3cm]{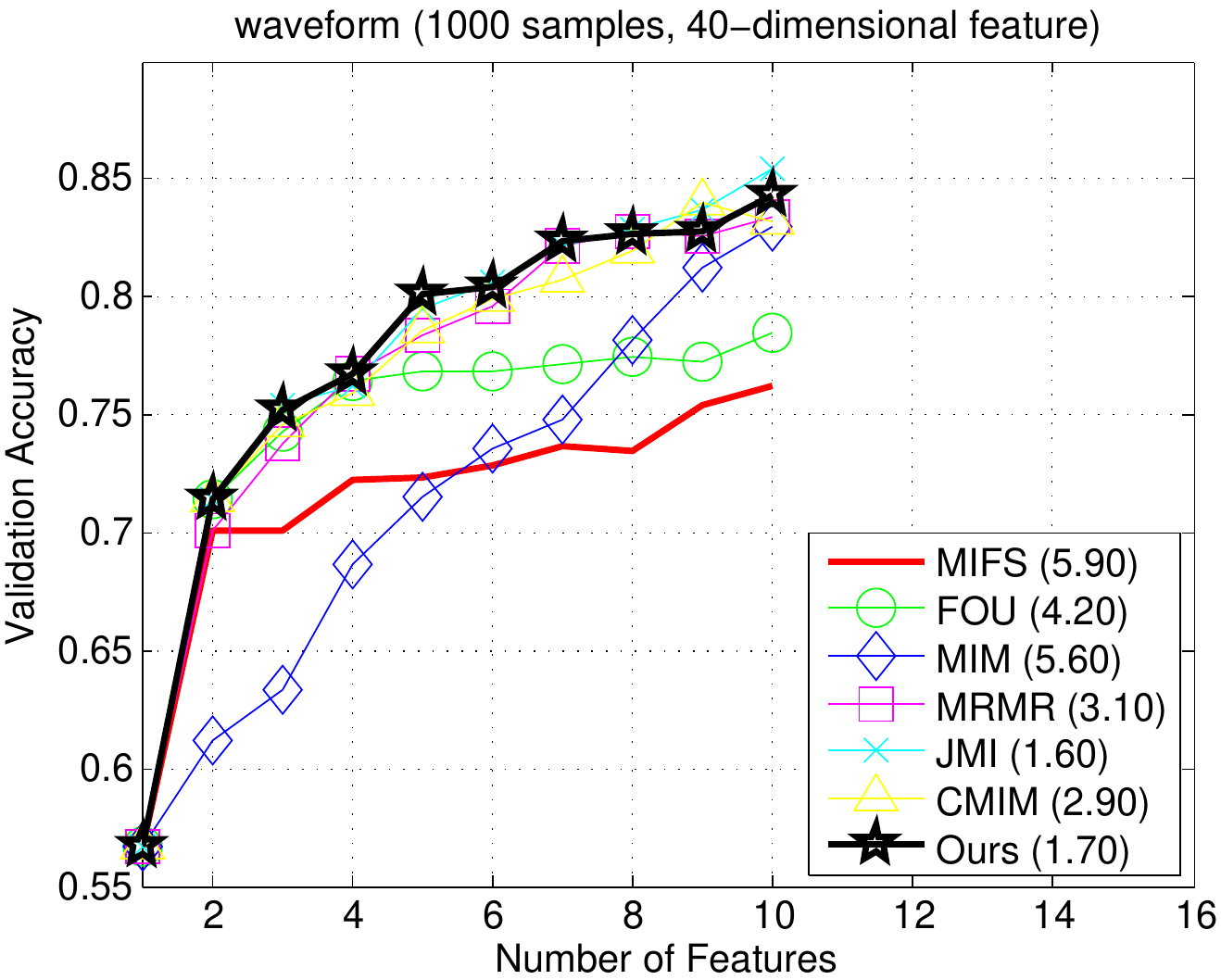}} \\
\subfigure[Lung] {\includegraphics[width=.23\textwidth,height=3cm]{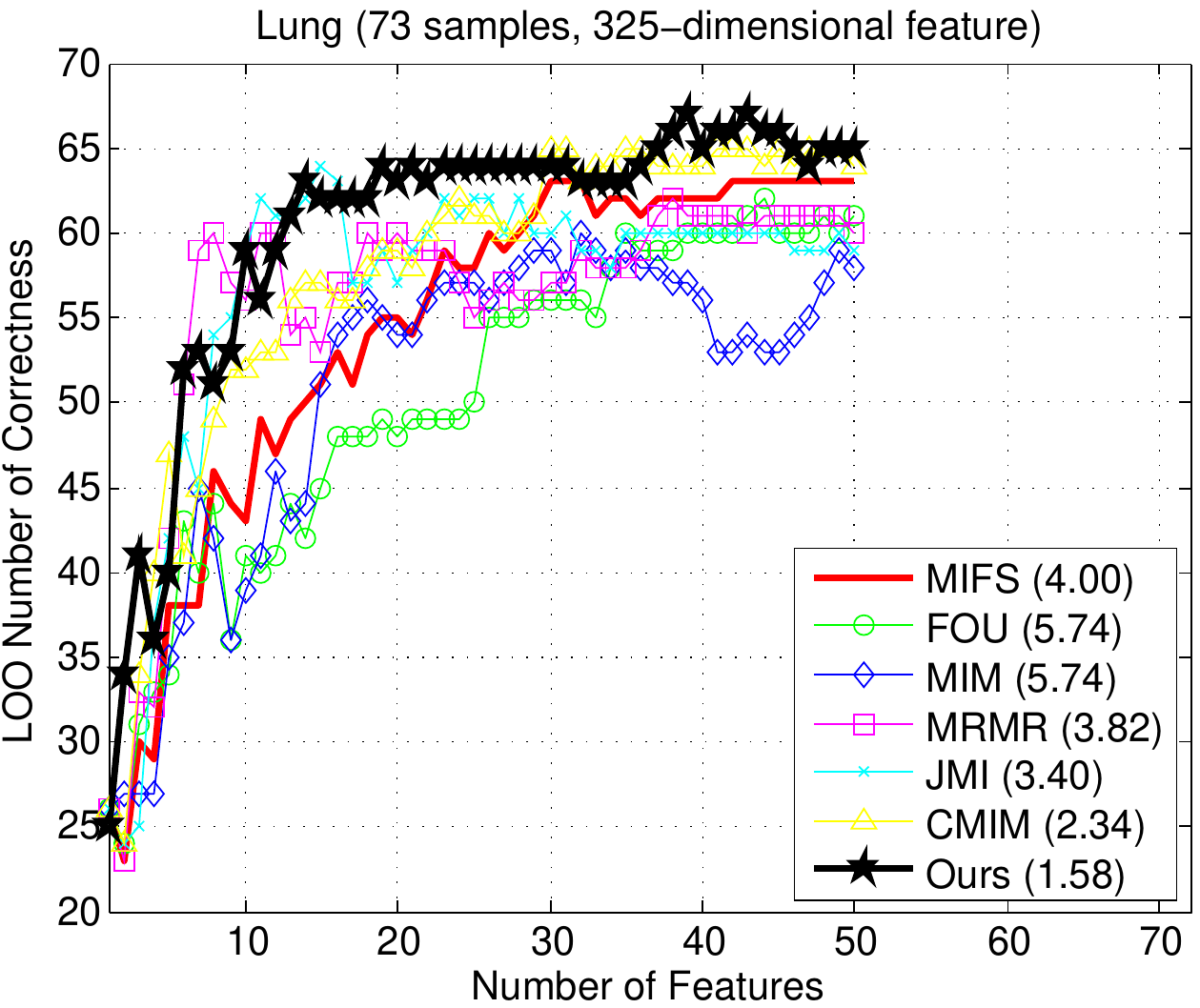}}
\subfigure[Lymph] {\includegraphics[width=.23\textwidth,height=3cm]{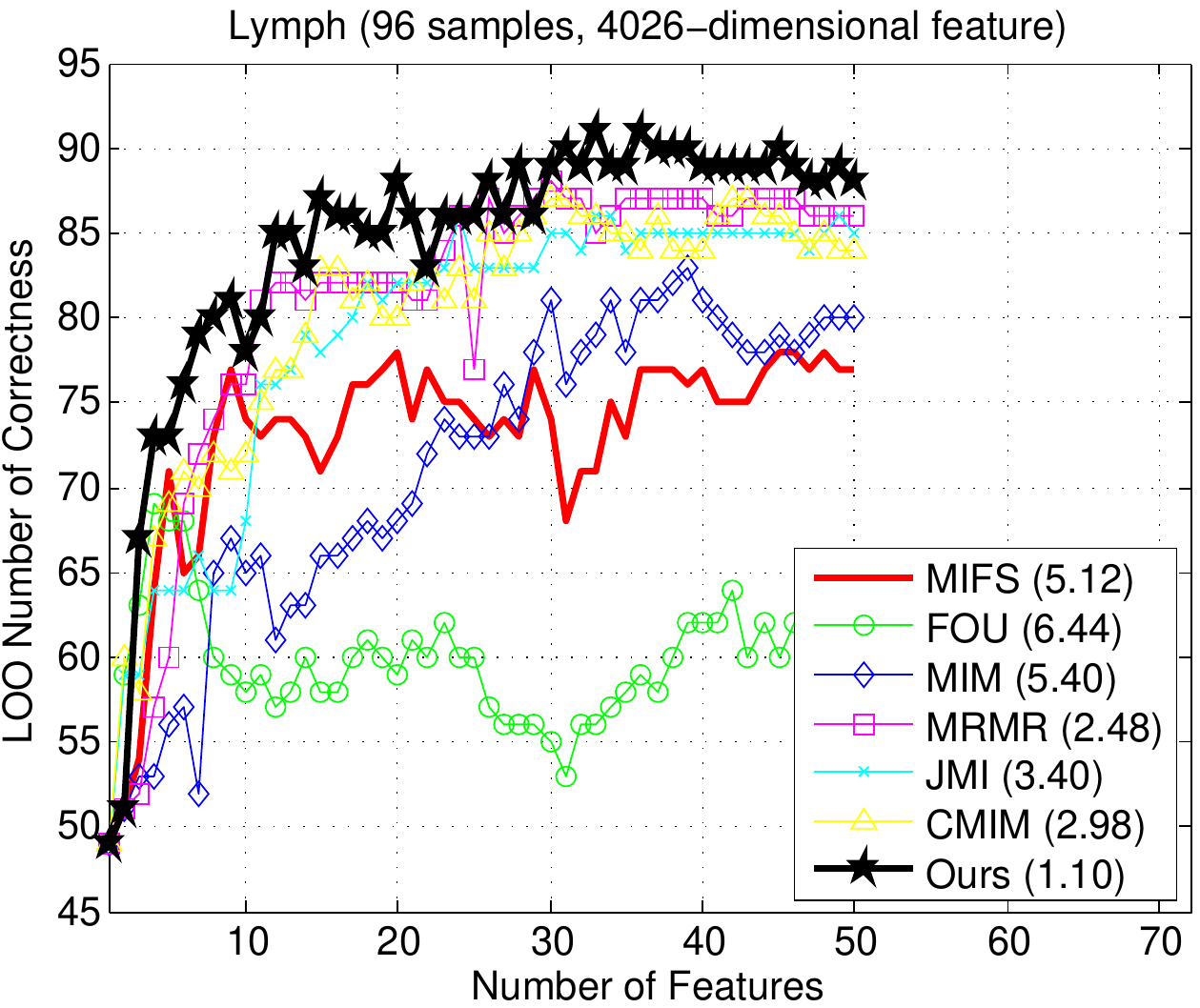}}
\subfigure[ORL] {\includegraphics[width=.23\textwidth,height=3cm]{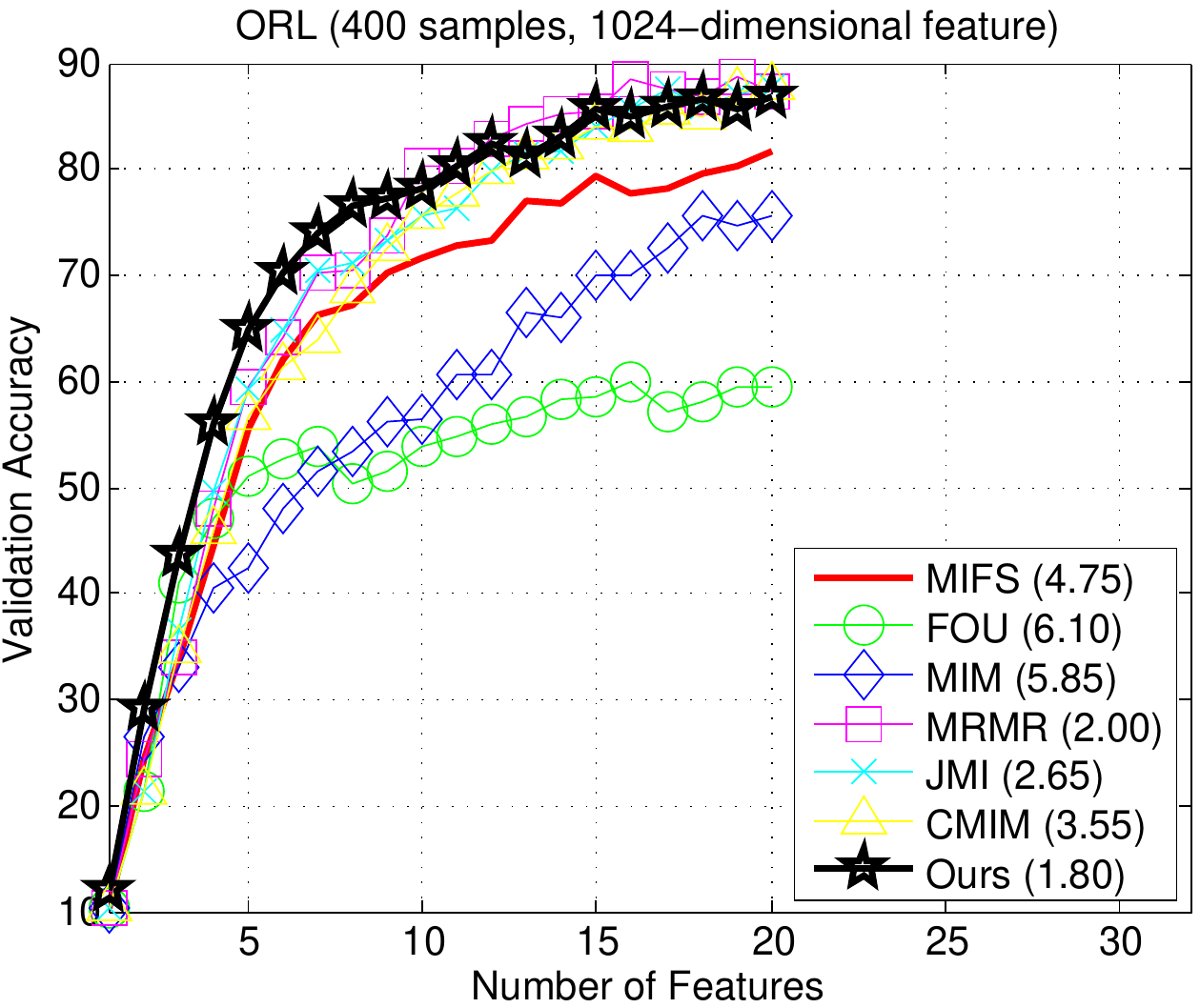}}
\subfigure[warpPIE10P] {\includegraphics[width=.23\textwidth,height=3cm]{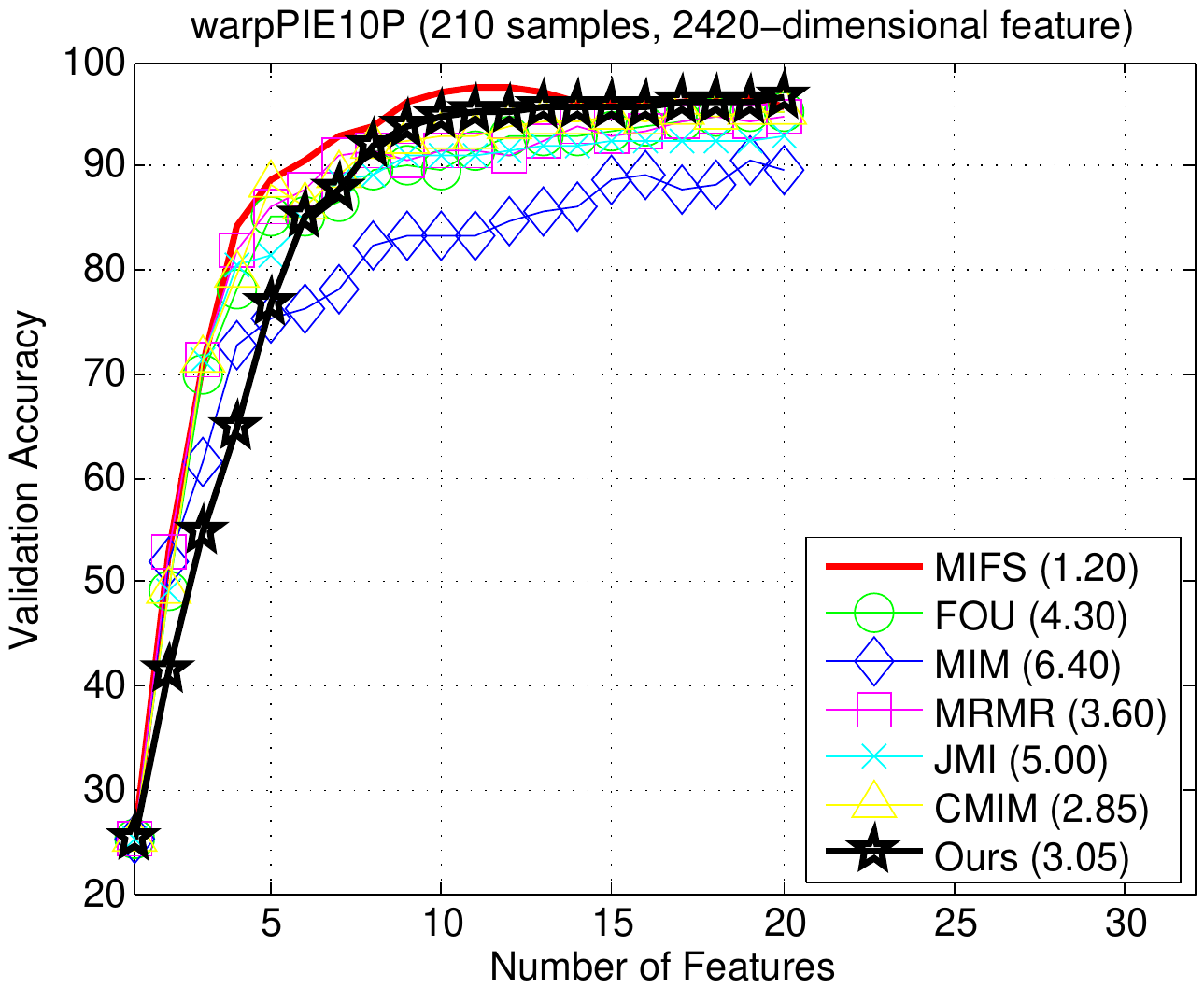}}
\caption{Validation accuracy or Leave-one-out (LOO) results on synthetic and real datasets. The number of samples and the feature dimensionality for each dataset are listed in the title. The value beside each method in the legend indicates the average rank in that dataset. Our method ($\alpha=0.6$) achieves the best or second-best performance in most datasets regardless of data characteristics.}
\label{fig:real_data_0_6}
\end{figure*}

\begin{table*}
\tiny
\centering
\caption{A summarization of different information-theoretic feature selection methods and their average ranks over different number of features in each dataset. The overall average ranks over different datasets are also reported. The best two performance in each dataset are marked with \textcolor{red}{red} and \textcolor{blue}{blue} respectively.}\label{lab:difference_summarization_0_6}
\begin{tabular}{ccccccccccc}\hline
 & Criteria & MADELON & breast & semeion & waveform & Lung & Lymph & ORL & PIE10P & average \\\hline
$\text{MIFS}$~\cite{battiti1994using} & $\mathbf{I}(X_{i_k};Y)-\beta\sum\limits_{l=1}^{k-1}\mathbf{I}(X_{i_k};X_{i_l})$ & $6.65$ & $6.30$ & $\color{red}{1.00}$ & $5.90$ & $4.00$ & $5.12$ & $4.75$ & $\color{red}{1.20}$ & $4.75$ \\
$\text{FOU}$~\cite{brown2009new} & $\mathbf{I}(X_{i_k};Y)-\sum\limits_{l=1}^{k-1}[\mathbf{I}(X_{i_k};X_{i_l})-\mathbf{I}(X_{i_k};X_{i_l}|Y)]$ & $\color{red}{1.85}$ & $5.00$ & $3.53$ & $4.20$ & $5.74$ & $6.44$ & $6.10$ & $4.30$ & $5.19$ \\
$\text{MIM}$~\cite{lewis1992feature} & $\mathbf{I}(X_{i_k};Y)$ & $3.95$ & $4.80$ & $6.60$ & $5.60$ & $5.74$ & $5.40$ & $5.85$ & $6.40$ & $6.06$ \\
$\text{MRMR}$~\cite{peng2005feature} & $\mathbf{I}(X_{i_k};Y)-\frac{1}{k-1}\sum\limits_{l=1}^{k-1}\mathbf{I}(X_{i_k};X_{i_l})$ & $5.85$ & $2.60$ & $4.00$ & $3.10$ & $3.82$ & $\color{blue}{2.48}$ & $\color{blue}{2.00}$ & $3.60$ & $3.75$ \\
$\text{JMI}$~\cite{yang2000data} & $\sum\limits_{l=1}^{k-1}\mathbf{I}(\{X_{i_k},X_{i_l}\};Y)$ & $2.50$ & $3.40$ & $5.20$ & $\color{red}{1.60}$ & $3.40$ & $3.40$ & $2.65$ & $5.00$ & $3.75$ \\
$\text{CMIM}$~\cite{fleuret2004fast} & $\min\limits_l\mathbf{I}(X_{i_k};Y|X_{i_l})$ & $3.70$ & $\color{blue}{2.10}$ & $2.93$ & $2.90$ & $\color{blue}{2.34}$ & $2.98$ & $3.55$ & $\color{blue}{2.85}$ & $\color{blue}{2.88}$ \\
$\text{Ours}~(\alpha=0.6)$ & $\mathbf{I}(\{X_{i_1},X_{i_2},\cdots,X_{i_k}\};Y)$ & $\color{blue}{2.45}$ & $\color{red}{1.80}$ & $\color{blue}{2.40}$ & $\color{blue}{1.70}$ & $\color{red}{1.58}$ & $\color{red}{1.10}$ & $\color{red}{1.80}$ & $3.05$ & $\color{red}{1.63}$ \\\hline
\end{tabular}
\end{table*}

\newpage
\subsection{Feature selection results with $\alpha=2$}\label{appendix_B}

We finally present feature selection results with $\alpha=2$. Note that, $\alpha=2$ is not recommended for the application of feature selection which involves classification. This is because classification uses a counting norm, attention should be paid on tails of the distribution or multiple modalities, thus values of $\alpha$ lower than $2$ are preferred. In fact, according to the quantitative evaluation results shown in Fig.~\ref{fig:real_data_2} and Table~\ref{lab:difference_summarization_2}, the performance of our method with $\alpha=2$ is decreased compared to the cases of $\alpha=1.01$ and $\alpha=0.6$.

\begin{figure*}[!t]
\setlength{\abovecaptionskip}{0.cm}
\setlength{\belowcaptionskip}{-0.0cm}
\centering
\subfigure[MADELON] {\includegraphics[width=.23\textwidth,height=3cm]{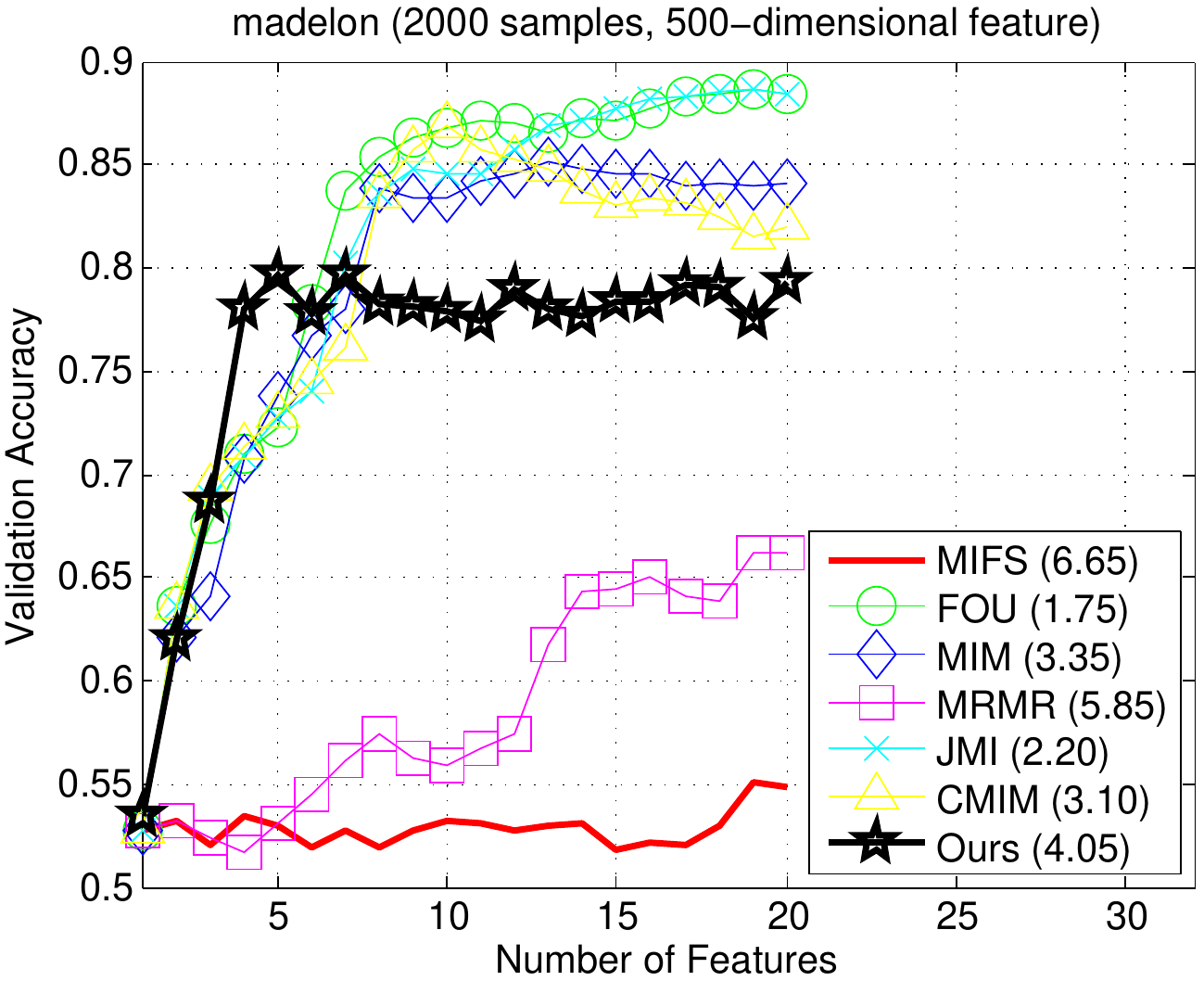}}
\subfigure[breast] {\includegraphics[width=.23\textwidth,height=3cm]{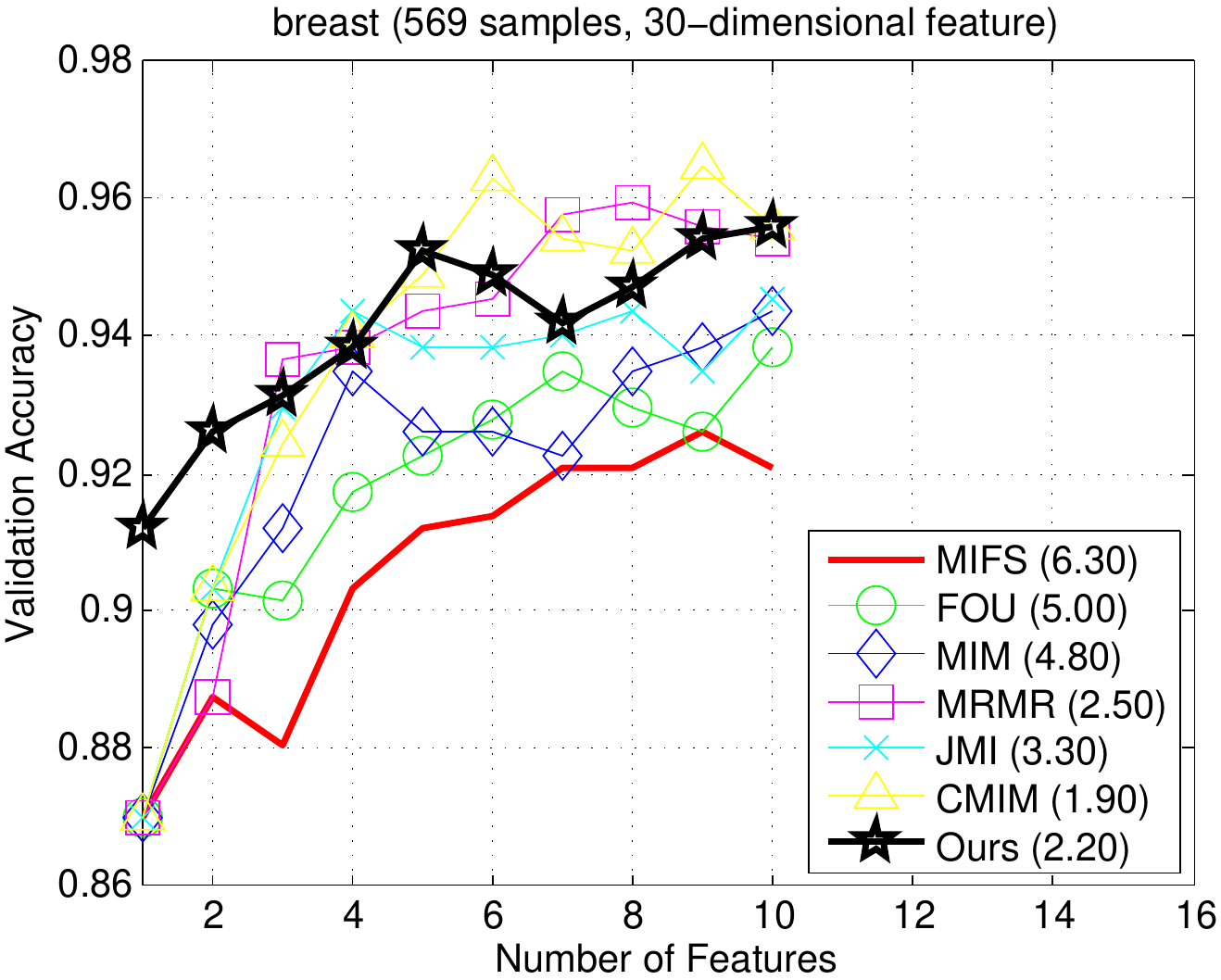}}
\subfigure[semeion] {\includegraphics[width=.23\textwidth,height=3cm]{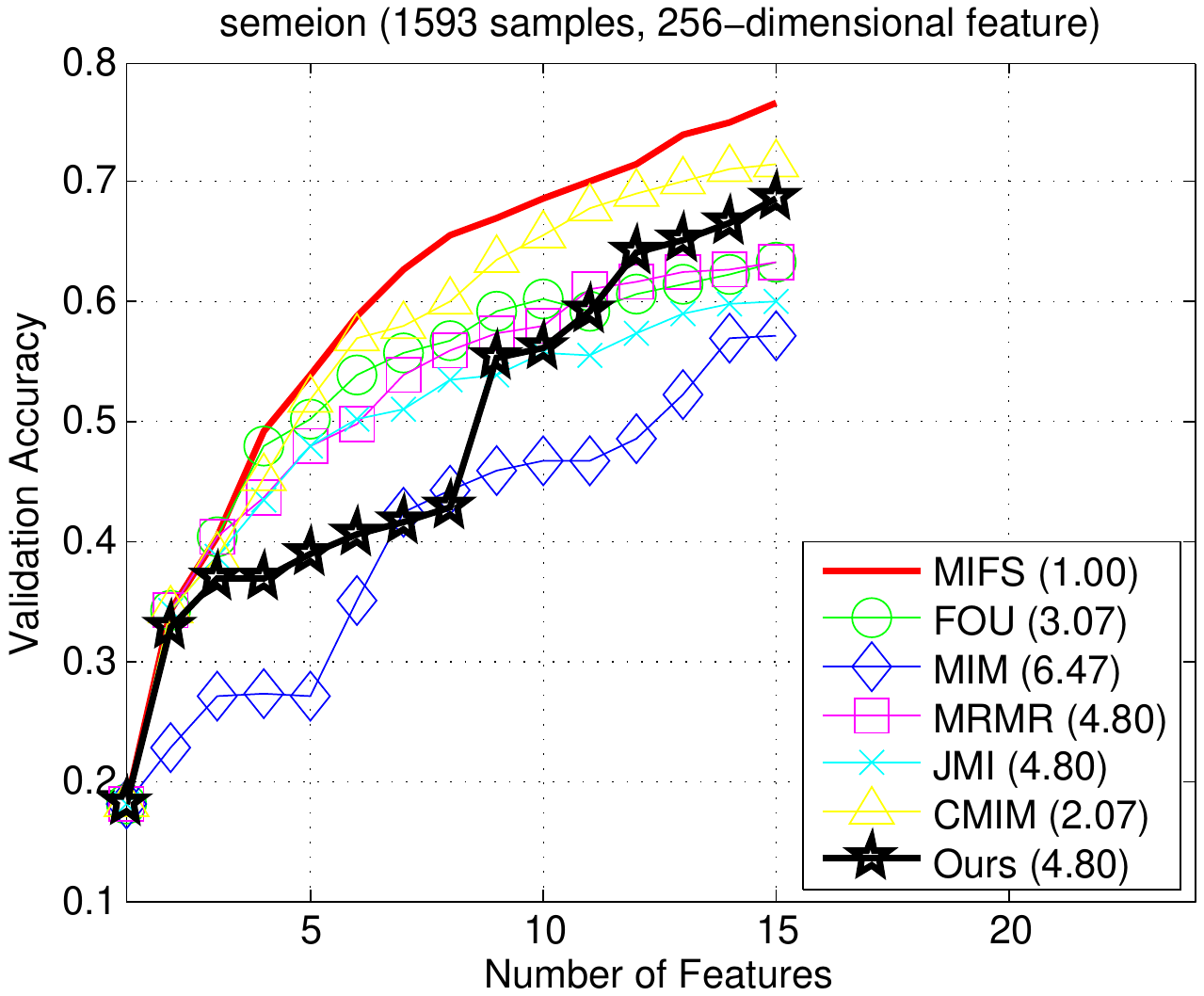}}
\subfigure[waveform] {\includegraphics[width=.23\textwidth,height=3cm]{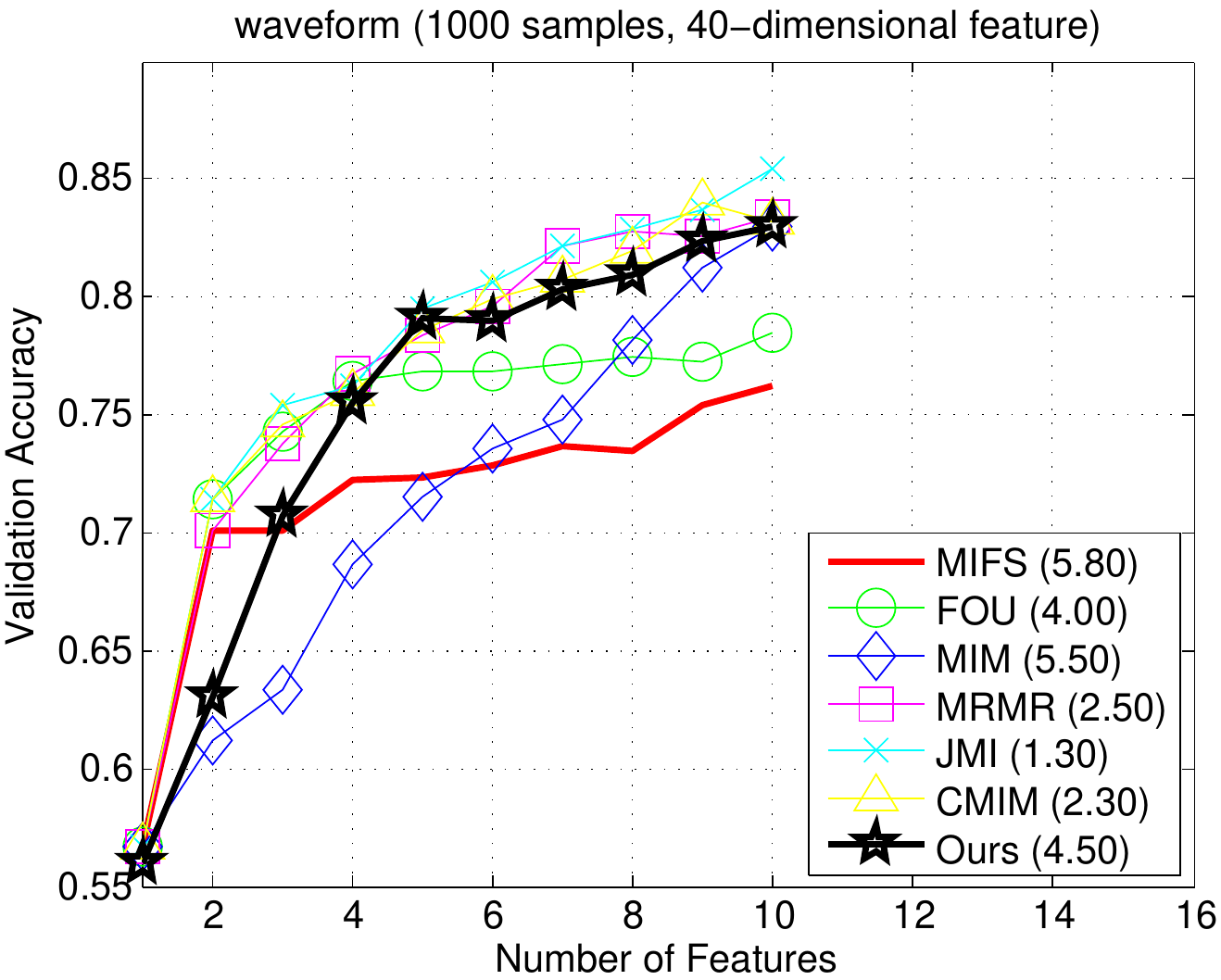}} \\
\subfigure[Lung] {\includegraphics[width=.23\textwidth,height=3cm]{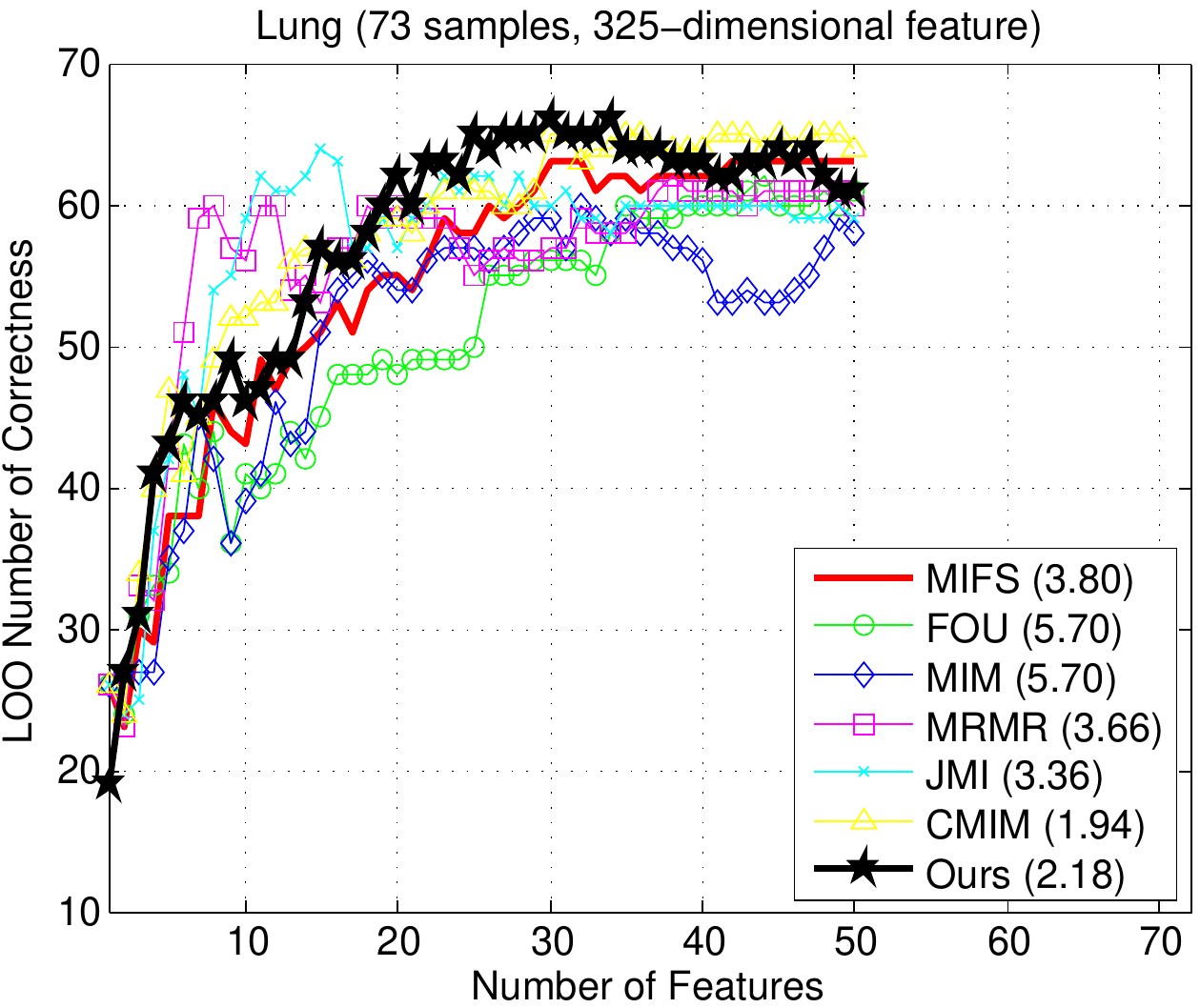}}
\subfigure[Lymph] {\includegraphics[width=.23\textwidth,height=3cm]{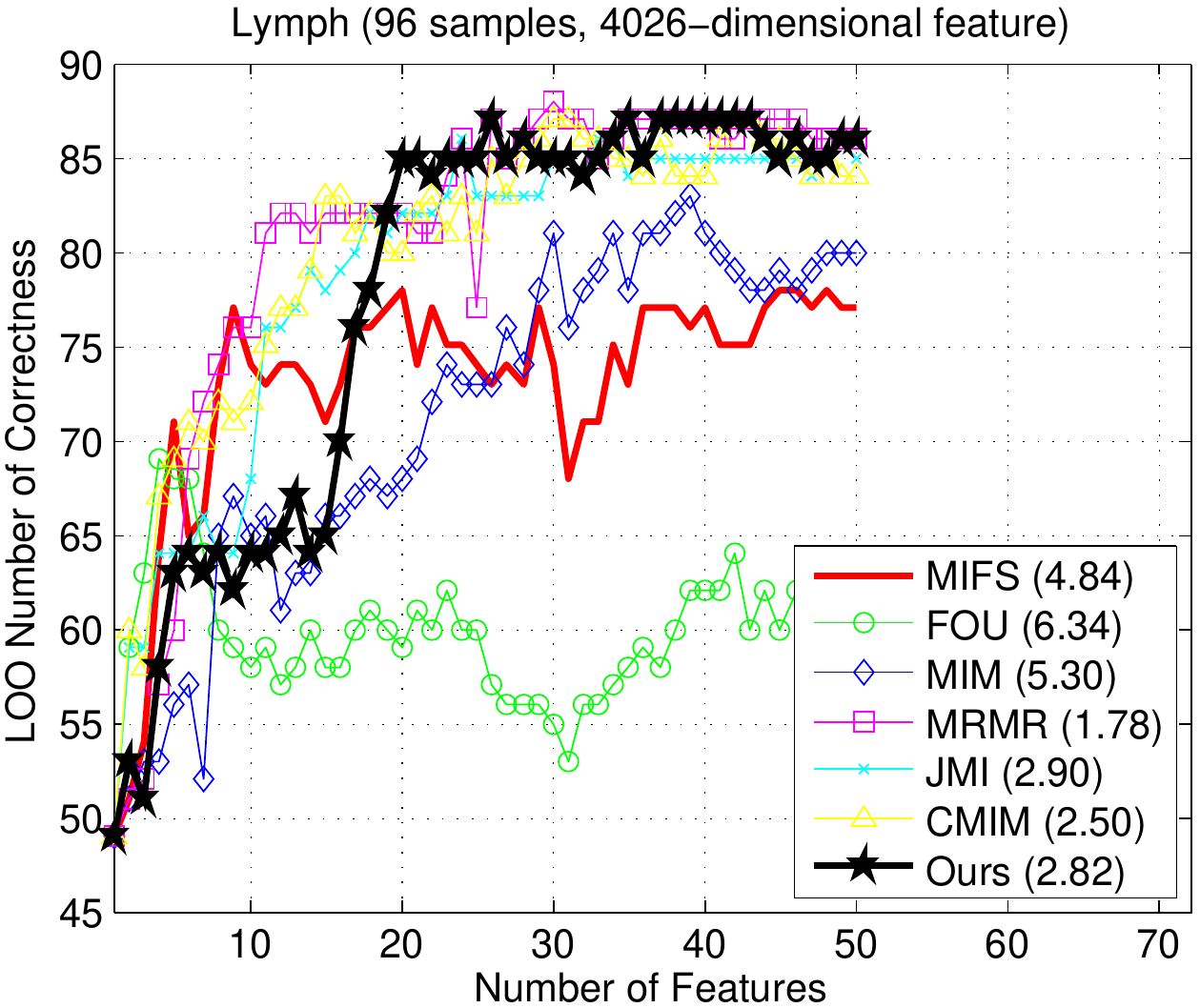}}
\subfigure[ORL] {\includegraphics[width=.23\textwidth,height=3cm]{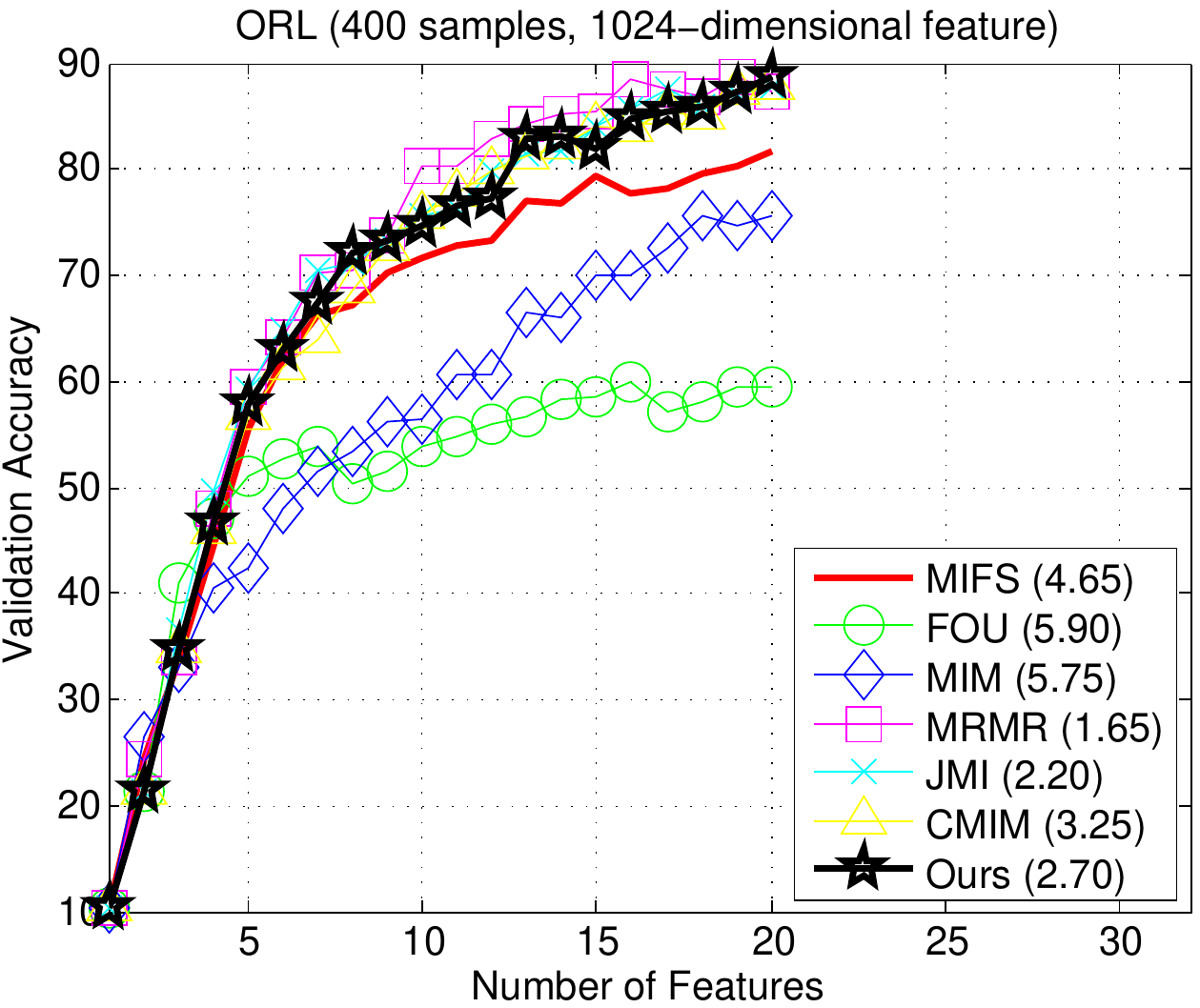}}
\subfigure[warpPIE10P] {\includegraphics[width=.23\textwidth,height=3cm]{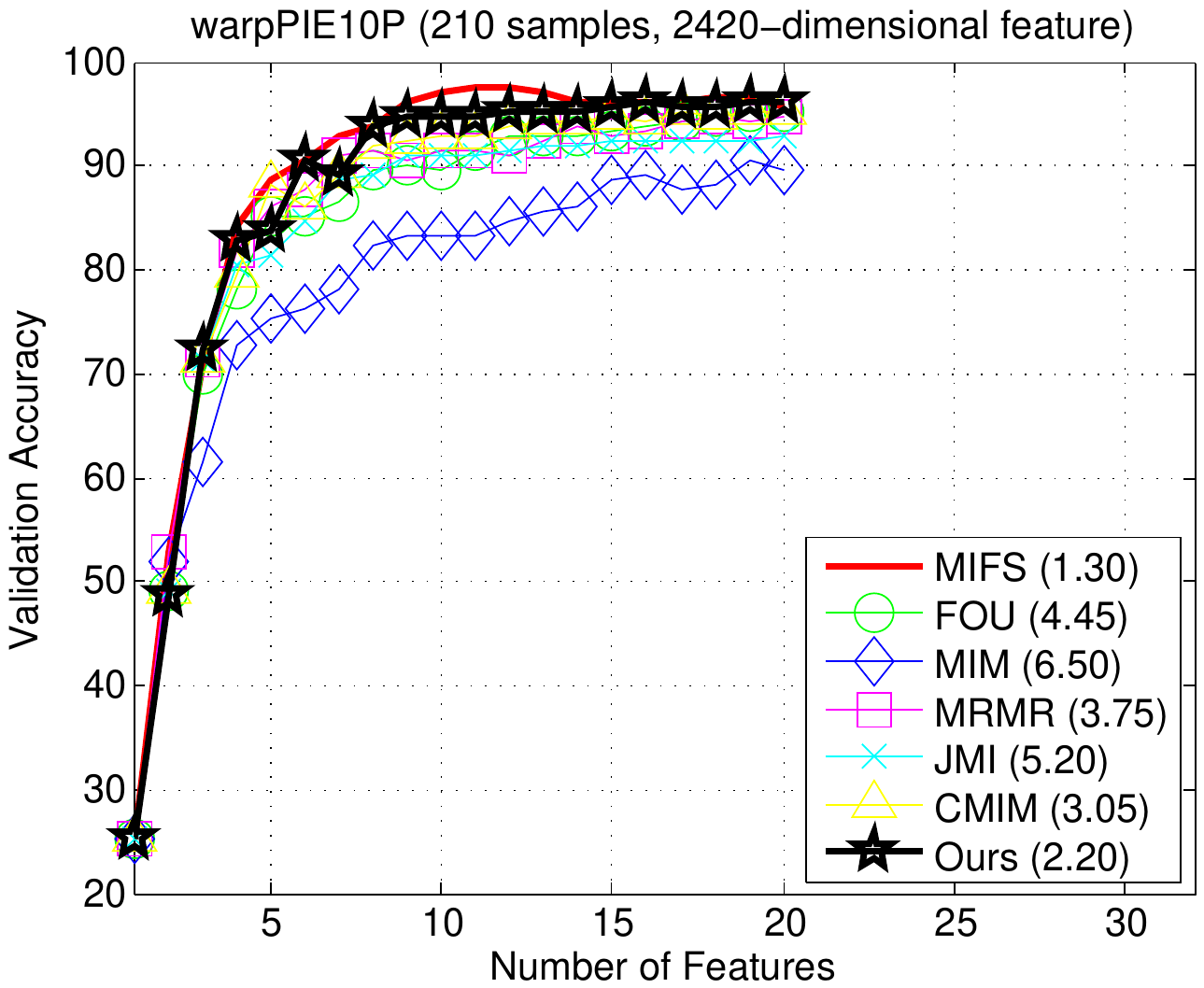}}
\caption{Validation accuracy or Leave-one-out (LOO) results on synthetic and real datasets. The number of samples and the feature dimensionality for each dataset are listed in the title. The value beside each method in the legend indicates the average rank in that dataset. The performance of our method ($\alpha=2$) is decreased in most datasets.}
\label{fig:real_data_2}
\end{figure*}

\begin{table*}
\tiny
\centering
\caption{A summarization of different information-theoretic feature selection methods and their average ranks over different number of features in each dataset. The overall average ranks over different datasets are also reported. The best two performance in each dataset are marked with \textcolor{red}{red} and \textcolor{blue}{blue} respectively.}\label{lab:difference_summarization_2}
\begin{tabular}{ccccccccccc}\hline
 & Criteria & MADELON & breast & semeion & waveform & Lung & Lymph & ORL & PIE10P & average \\\hline
$\text{MIFS}$~\cite{battiti1994using} & $\mathbf{I}(X_{i_k};Y)-\beta\sum\limits_{l=1}^{k-1}\mathbf{I}(X_{i_k};X_{i_l})$ & $6.65$ & $6.30$ & $\color{red}{1.00}$ & $5.80$ & $3.80$ & $4.84$ & $4.65$ & $\color{red}{1.30}$ & $4.75$ \\
$\text{FOU}$~\cite{brown2009new} & $\mathbf{I}(X_{i_k};Y)-\sum\limits_{l=1}^{k-1}[\mathbf{I}(X_{i_k};X_{i_l})-\mathbf{I}(X_{i_k};X_{i_l}|Y)]$ & $\color{red}{1.75}$ & $5.00$ & $3.07$ & $4.00$ & $5.70$ & $6.34$ & $5.90$ & $4.45$ & $4.94$ \\
$\text{MIM}$~\cite{lewis1992feature} & $\mathbf{I}(X_{i_k};Y)$ & $3.35$ & $4.80$ & $6.47$ & $5.50$ & $5.70$ & $5.30$ & $5.75$ & $6.50$ & $5.94$ \\
$\text{MRMR}$~\cite{peng2005feature} & $\mathbf{I}(X_{i_k};Y)-\frac{1}{k-1}\sum\limits_{l=1}^{k-1}\mathbf{I}(X_{i_k};X_{i_l})$ & $5.85$ & $2.50$ & $4.80$ & $2.50$ & $3.66$ & $\color{red}{1.78}$ & $\color{red}{1.65}$ & $3.75$ & $\color{blue}{3.31}$ \\
$\text{JMI}$~\cite{yang2000data} & $\sum\limits_{l=1}^{k-1}\mathbf{I}(\{X_{i_k},X_{i_l}\};Y)$ & $\color{blue}{2.20}$ & $3.30$ & $4.80$ & $\color{red}{1.30}$ & $3.36$ & $2.90$ & $\color{blue}{2.20}$ & $5.20$ & $3.44$ \\
$\text{CMIM}$~\cite{fleuret2004fast} & $\min\limits_l\mathbf{I}(X_{i_k};Y|X_{i_l})$ & $3.10$ & $\color{red}{1.90}$ & $\color{blue}{2.07}$ & $\color{blue}{2.30}$ & $\color{red}{1.94}$ & $\color{blue}{2.50}$ & $3.25$ & $3.05$ & $\color{red}{2.25}$ \\
$\text{Ours}~(\alpha=2)$ & $\mathbf{I}(\{X_{i_1},X_{i_2},\cdots,X_{i_k}\};Y)$ & $4.05$ & $\color{blue}{2.20}$ & $4.80$ & $4.50$ & $\color{blue}{2.18}$ & $2.82$ & $2.70$ & $\color{blue}{2.20}$ & $3.38$ \\\hline
\end{tabular}
\end{table*}

\begin{IEEEbiography}
	[{\includegraphics[width=1in,height=1.25in,clip,keepaspectratio]{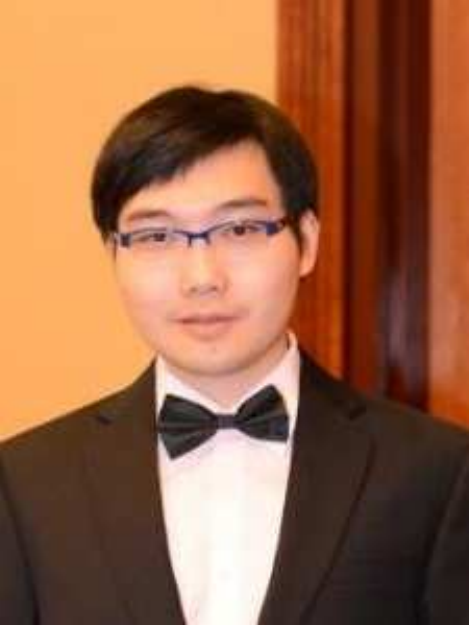}}]{Shujian Yu}
	(SM'17) received his Ph.D. degree in Electrical and Computer Engineering with a Ph.D. minor in Statistics at the University of Florida in 2019. He earned his B.S. degree from the School of Electronic Information and Communications at the Huazhong University of Science and Technology in 2013. Shujian was a Data Mining Research Scientist Intern at the Bosch Research and Technology Center in Palo Alto, CA, in summer 2016, and a Machine Learning Research Scientist Intern at the Nokia Bell Labs in Murray Hill, NJ, in summer 2017. He joined the NEC Labs Europe in Heidelberg, Germany, in August 2019, as a Machine Learning Research Scientist. His research interests include topics in machine learning, signal processing, and information theory.
\end{IEEEbiography}

\begin{IEEEbiography}
	[{\includegraphics[width=1in,height=1.25in,clip,keepaspectratio]{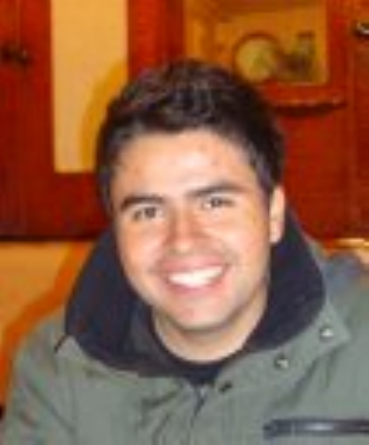}}]{Luis Gonzalo S\'{a}nchez Giraldo}
	received the BS degree in Electronics Engineering and the MEng degree in Industrial Automation from Universidad Nacional de Colombia in 2005 and 2008, respectively, and the PhD degree in Electrical and Computer Engineering from the University of Florida in 2012. He is currently a postdoctoral associate in the Department of Computer Science at the University of Miami. Between 2004 and 2008, he was appointed as a research assistant at the Control and Digital Signal Processing Group (GCPDS) at Universidad Nacional de Colombia. During his PhD studies, he was a research assistant at the Computational NeuroEngineering Laboratory (CNEL) at the University of Florida. His main research interests are in machine learning, signal processing, and computational neuroscience.
\end{IEEEbiography}

\begin{IEEEbiography}
    [{\includegraphics[width=1in,height=1.25in,clip,keepaspectratio]{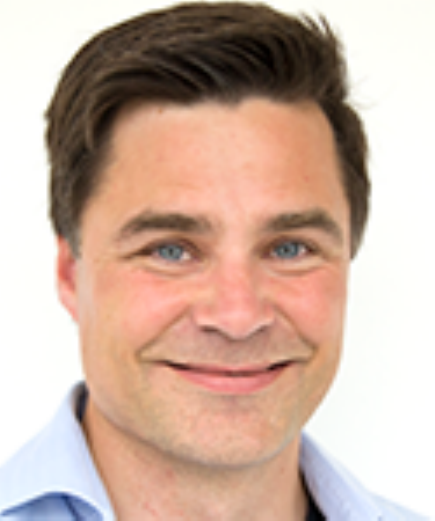}}]{Robert Jenssen}
    (M'02) received the PhD (Dr. Scient.) in Electrical Engineering from the University of Troms{\o}, in 2005. Currently, he is a Professor and Head of the UiT Machine Learning Group at UiT - The Arctic University of Norway, Troms{\o}, Norway. Jenssen received the Honorable Mention for the 2003 Pattern Recognition Journal Best Paper Award, the 2005 IEEE ICASSP Outstanding Student Paper Award, the 2007 UiT Young Investigator Award, the 2013 IEEE Geoscience and Remote Sensing Society Letters Best Paper Award, and the 2017 Scandinavian Conference on Image Analysis Best Student Paper Award (as supervisor). His paper, ``Kernel Entropy Component Analysis," was a featured paper of the IEEE Transactions on Pattern Analysis and Machine Intelligence. Jenssen is an adjunct professor with the Norwegian Computing Center in Oslo. He was previously a senior researcher (20$\%$) at the Norwegian Center on E-Health Research. He was Guest Researcher with the Technical University of Denmark, Kongens Lyngby, Denmark, from 2012 to 2013, with the Technical University of Berlin, Berlin, Germany, from 2008 to 2009, and with the University of Florida, Gainesville, FL, USA, spring 2018, spring 2004, and from Sept. 2002 to July 2003. Jenssen serves on the IEEE Technical Committee on Machine Learning for Signal Processing, he is on the IAPR Governing Board, and he is an Associate Editor for the journal Pattern Recognition. Jenssen is the General Chair of the annual Northern Lights Deep Learning Workshop - NLDL.
\end{IEEEbiography}

\begin{IEEEbiography}
	[{\includegraphics[width=1in,height=1.25in,clip,keepaspectratio]{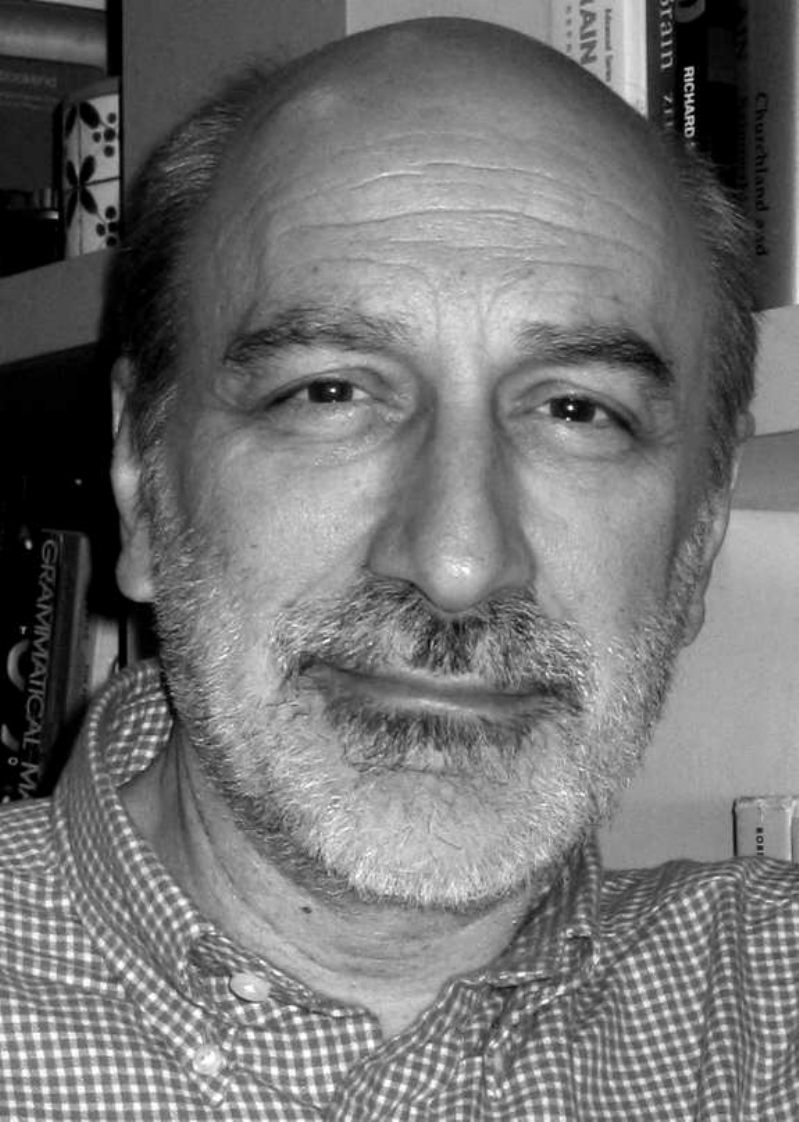}}]{Jos\'{e} C. Pr\'{i}ncipe}
	(M'83-SM'90-F'00) is the BellSouth and Distinguished Professor of Electrical and Biomedical Engineering at the University of Florida, and the Founding Director of the Computational NeuroEngineering Laboratory (CNEL). His primary research interests are in advanced signal processing with information theoretic criteria and adaptive models in reproducing kernel Hilbert spaces (RKHS), with application to brain-machine interfaces (BMIs). Dr.  Pr\'{i}ncipe is a Fellow of the IEEE, ABME, and AIBME. He is the past Editor in Chief of the IEEE Transactions on Biomedical Engineering, past Chair of the Technical Committee on Neural Networks of the IEEE Signal Processing Society, Past-President of the International Neural Network Society, and a recipient of the IEEE EMBS Career Award and the IEEE Neural Network Pioneer	Award.
\end{IEEEbiography}

\end{document}